\documentclass[a4paper,UKenglish,cleveref, autoref, thm-restate,numberwithinsect]{lipics-v2021}
\nolinenumbers

\hideLIPIcs
\usepackage[mathscr]{eucal}
\usepackage{xspace}

\newcommand{\bigO}{\mathcal{O}}
\newcommand{\G}{\mathcal{G}}
\newcommand{\reg}{\lambda}

\newcommand{\phisym}{\phi_{\rm sym}}
\newcommand{\fsym}{f_{\rm sym}}
\newtheorem{hypothesis}[theorem]{Hypothesis}

\newcommand{\cspk}{\textsf{CSP}_{k}}
\newcommand{\maxcspk}{\textsf{maxCSP}_{k}}
\newcommand{\maxcsp}{\textsf{maxCSP}}
\newcommand{\mincspk}{\textsf{minCSP}_{k}}
\newcommand{\Fam}{\mathcal{F}}

\newcommand{\nand}{\textsf{NAND}\xspace}
\newcommand{\AND}{\textsf{AND}\xspace}
\newcommand{\xor}{\textsf{XOR}\xspace}
\newcommand{\OR}{\textsf{OR}\xspace}
\newcommand{\impl}{\textsf{Implication}\xspace}
\newcommand*\tageq{\refstepcounter{equation}\tag{\theequation}}

\usepackage{mathtools}

\usepackage{cancel}
\usepackage{textgreek}
\usepackage{apxproof}

\title{The Role of Regularity in (Hyper-)Clique Detection and Implications for Optimizing Boolean CSPs} 
\titlerunning{The Role of Regularity in (Hyper-)Clique Detection} 

\author{Nick Fischer}{INSAIT, Sofia University ``St.\ Kliment Ohridski''}{nickrobinfischer@gmail.com}{https://orcid.org/0009-0001-0909-3296}{Partially funded by the Ministry of Education and Science of Bulgaria’s support for INSAIT, Sofia University ``St.\ Kliment Ohridski'' as part of the Bulgarian National Roadmap for Research Infrastructure. Parts of this work were done while the author was at Weizmann Institute of Science.}
\author{Marvin Künnemann}{Karlsruhe Institute of Technology}{marvin.kuennemann@kit.edu}{}{}
\author{Mirza Redžić}{Karlsruhe Institute of Technology}{mirza.redzic@kit.edu}{https://orcid.org/0009-0001-7509-1686}{}
\author{Julian Stieß}{Karlsruhe Institute of Technology}{julian.stiess@kit.edu}{https://orcid.org/0009-0002-7318-2645}{}
\funding{Research supported by the Deutsche Forschungsgemeinschaft (DFG, German Research Foundation) – 462679611.}
\authorrunning{N. Fischer, M. Künnemann, M. Redžić, J. Stieß} 

\Copyright{{Nick Fischer, Marvin Künnemann, Mirza Redžić and Julian Stieß}} 



\category{Track A: Algorithms, Complexity and Games} 

\relatedversion{} 

\EventEditors{ Keren Censor-Hillel, Fabrizio Grandoni, Joel Ouaknine, and Gabriele Puppis}
\EventNoEds{4}
\EventLongTitle{52nd International Colloquium on Automata, Languages, and Programming (ICALP 2025)}
\EventShortTitle{ICALP 2025}
\EventAcronym{ICALP}
\EventYear{2025}
\EventDate{July 8--11, 2025}
\EventLocation{Aarhus, Denmark}
\EventLogo{}
\SeriesVolume{334}
\ArticleNo{120}
\begin{document}

\maketitle

\begin{abstract}
Is detecting a $k$-clique in $k$-partite \emph{regular} (hyper-)graphs as hard as in the general case? Intuition suggests yes, but proving this---especially for hypergraphs---poses notable challenges. Concretely, we consider a strong notion of regularity in $h$-uniform hypergraphs, where we essentially require that any subset of at most $h-1$ is incident to a uniform number of hyperedges. Such notions are studied intensively in the combinatorial block design literature. We show that any $f(k)n^{g(k)}$-time algorithm for detecting $k$-cliques in such graphs transfers to an $f'(k)n^{g(k)}$-time algorithm for the general case, establishing a fine-grained equivalence between the $h$-uniform hyperclique hypothesis and its natural regular analogue. 

Equipped with this regularization result, we then fully resolve the fine-grained complexity of optimizing Boolean constraint satisfaction problems over assignments with $k$ non-zeros. Our characterization depends on the maximum degree $d$ of a constraint function. Specifically, if $d\le 1$, we obtain a linear-time solvable problem, if $d=2$, the time complexity is essentially equivalent to $k$-clique detection, and if $d\ge 3$ the problem requires exhaustive-search time under the 3-uniform hyperclique hypothesis. To obtain our hardness results, the regularization result plays a crucial role, enabling a very convenient approach when applied carefully.
We believe that our regularization result will find further applications in the future.
\end{abstract}

\section{Introduction}
In the design of algorithms and hardness reductions, it is often helpful to assume that the involved instances satisfy certain \emph{regularity} conditions. To name a concrete example: For his celebrated Set Cover inapproximability result, Feige~\cite{Feige98} first argues that the 3SAT-5 problem (the highly regular special case of SAT in which every clause contains precisely 3 variables and each variable occurs in precisely 5 clauses) is hard to approximate, and then gives a reduction from 3SAT-5 to Set Cover. We informally refer to the first reduction as a \emph{regularization}. Another notable example of a regularization in parameterized complexity is the $W[1]$-hardness of $k$-Clique detection in \emph{regular} graphs. In fact, this is the first (non-trivial) reduction in the textbook by Cygan et al.~\cite{CyganFKLMPPS15}, and has been independently observed for their respective applications in at least three different works~\cite{Marx06,Cai08,MathiesonS08}.

Conversely, obtaining regularization results can also have \emph{algorithmic} implications. For instance, Zamir~\cite{Zamir23} devises improved exact algorithms for a wide array of problems including graph coloring on \emph{regular} graphs (in fact, even for \emph{almost-regular} graphs). Designing regularizations for these problems would therefore imply breakthroughs for general instances.

In this work we focus on the $k$-Clique and $k$-Hyperclique problems from a fine-grained perspective. Our main result is a surprisingly intricate fine-grained regularization reduction, i.e., we prove that both problems are as hard on regular (hyper-)graphs as they are on general (hyper-)graphs. Equipped with this, we are able to resolve the fine-grained complexity of optimizing Boolean CSPs over weight-$k$ solutions; this, in fact, had been the motivating question for this study. Below, we detail these two contributions separately in Sections~\ref{sec:intro-hyperclique} and~\ref{sec:intro-optCSPs}. 

\subsection{Regularizing (Hyper-)clique Detection}
\label{sec:intro-hyperclique}
\subparagraph*{Clique Detection.}
The \emph{$k$-Clique} problem is to decide if a $k$-partite graph $G = (V_1, \dots, V_k, E)$ contains a \emph{$k$-clique}, i.e., $k$ nodes $v_1 \in V_1, \dots, v_k \in V_k$ such that any pair $(v_i, v_j)$ forms an edge in~$G$.\footnote{Often, the problem is stated for general instead of $k$-partite graphs. However, these formulations are well-known to be equivalent using color-coding techniques, and the restriction to $k$-partite instances often helps in designing reductions.} While the trivial algorithm takes time $\bigO(n^k)$, Ne\v{s}et\v{r}il and Poljak~\cite{NesetrilP85} devised an algorithm running in time $\bigO(n^{(\omega/3) k})$ where $\omega < 2.3714$ denotes the matrix multiplication exponent~\cite{AlmanDWXXZ25}.\footnote{More precisely, the running time bound $\bigO(n^{(\omega/3) k})$ applies only when $k$ is divisible by $3$. For general $k$, the running time is~\smash{$\bigO(n^{\omega(\lfloor k/3 \rfloor, \lceil k/3 \rceil, \lceil (k-1)/3\rceil)})$}, where $\omega(a, b, c)$ is the exponent of rectangular matrix multiplication (i.e., \smash{$\bigO(n^{\omega(a, b, c)})$} is the running time of multiplying an~\smash{$n^a \times n^b$} matrix by an~\smash{$n^b \times n^c$} matrix).} The \emph{$k$-Clique Hypothesis} postulates that this algorithm is essentially optimal (see Section~\ref{sec:preliminaries}). As it generalizes the popular \emph{Triangle} hypothesis (i.e., the special case~\makebox{$k = 3$}, stating that the complexity of Triangle Detection is $n^{\omega\pm o(1)}$), this hardness assumption entails fine-grained lower bounds for an extensive list of problems (see e.g.~\cite{Vassilevska18}), but also for some important lower bounds beyond the Triangle case~\cite{AbboudBW15a,Chang16,BringmannW17}.

\subparagraph*{Hyperclique Detection.}
The natural generalization to \emph{$h$-uniform hypergraphs}, the \emph{$(k, h)$-Hyperclique} problem, is to decide if a given $k$-partite $h$-uniform hypergraph $G = (V_1, \dots, V_k, E)$ contains a \emph{$k$-hyperclique}, i.e., $k$ nodes $v_1 \in V_1, \dots, v_k \in V_k$ such that all $h$-tuples $x_{i_1}, \dots, x_{i_h}$ form a hyperedge in $G$. In contrast to the $k$-Clique problem in graphs (which coincides with $(k, 2)$-Hyperclique), the $(k, h)$-Hyperclique problem for $h \geq 3$ is not known to admit any significant improvement over the brute-force search time $\bigO(n^k)$. This led to the \emph{($h$-Uniform) $k$-Hyperclique Hypothesis} postulating that time $n^{k-o(1)}$ is necessary in $h$-uniform hypergraphs~\cite{LincolnWW18}. Besides being a natural problem in itself, the relevance of the Hyperclique problem is also due to its close connection to the MAX-3-SAT problem~\cite{LincolnWW18} and due to a recent surge in strong (often matching) conditional lower bounds based on the Hyperclique conjecture~\cite{AbboudBDN18,LincolnWW18,BringmannFK19,CarmeliZBKS20,KunnemannM20,AnGIJKN21,BringmannS21,DalirrooyfardW22,Kunnemann22,GorbachevK23,Zamir23}.

\subparagraph*{(Hyper-)clique Regularization.}
The starting point of this paper is the following question: Do these hypotheses imply the same time bounds for clique detection in \emph{regular} (hyper)graphs? While it is clear what we mean by \emph{regular graph}, it is less clear how to define a regular $h$-uniform hypergraph. The perhaps simplest definition requires each node to be incident to the same number $r_1$ of hyperedges. However, this notion turns out to be insufficient for our main application. Instead, we consider a significantly stronger variant, in which \emph{every subset of $\ell<h$ vertices}, each taken from a different part of the $k$-partition, is incident to the same number $r_\ell$ of hyperedges (see Section~\ref{sec:preliminaries} for details). This yields essentially the strongest form of regularity that we may wish for in a hypergraph. The existence of such structures is well-studied in combinatorial block design theory; e.g., $k$-partite $3$-uniform hypergraphs that are regular in our sense correspond precisely to group divisible designs with $k$ groups, see~\cite{ColbournD06}.

Our first main contribution is to establish equivalence of the 3-Uniform Hyperclique Hypothesis for general hypergraphs to its special case of regular $k$-partite hypergraphs in the above strong sense.

\begin{theorem}[Hyperclique Regularization, Informal]\label{thm:main1}
Let $h \geq 2$. The $(k, h)$-Hyperclique problem is solvable in time $f(k) \cdot n^{g(k)}$ on $k$-partite $h$-uniform hypergraphs if and only if it is solvable in time $f'(k) \cdot n^{g(k)}$ on \emph{regular} $k$-partite $h$-uniform hypergraphs.
\end{theorem}

We are confident that this theorem will be useful in a black-box fashion for future hardness results based on the $k$-Clique or $k$-Hyperclique hypotheses. Moreover, we extend this equivalence also for \emph{counting} $k$-hypercliques, see Section~\ref{sec:regularizing-hypergraphs}.

\medskip
Let us discuss Theorem~\ref{thm:main1}: While the statement might feel ``obviously true'', to our surprise, proving it turns out to be anything but straightforward. Even for graphs ($h = 2$), we are not aware of a previous proof of this fact. The parameterized reductions in~\cite{Cai08,MathiesonS08,BrandesHK16,CyganFKLMPPS15} are generally not efficient enough as they reduce arbitrary $n$-vertex graphs $G$ with maximum degree $\Delta$ to regular graphs with $\Omega(\Delta n)$ vertices. As $\Delta$ can be as large as $\Theta(n)$, this blow-up is inadmissible for our fine-grained perspective (the resulting conditional lower bounds would have the exponent of $n$ halved).

Having said that, it is possible to regularize graphs by only adding $\bigO(n)$ additional nodes. Specifically, one approach is to increase the degree $\deg(v)$ of each node $v$ by $\Delta-\deg(v)$. To this end we can add $\bigO(n)$ dummy nodes and add an appropriate number of edges between the original nodes and the freshly introduced dummy nodes (while ensuring not to create new $k$-cliques). Afterwards it still remains to regularize the new dummy nodes, which can for instance be solved by balancing the new edges appropriately. The details are tedious and not very insightful.

The real challenge arises for hypergraphs. Recall that the strong regularity guarantee for, say, 3-uniform hypergraphs, is that any pair of vertices (in different vertex parts) must be contained in the same number of hyperedges. Thus, in order to generalize the previous approach we would have to regularize not only the individual dummy nodes, but also all \emph{pairs} of (original, dummy) and (dummy, dummy) nodes. For this reason we essentially make no progress by adding dummy nodes. Put differently, regularizing graphs is a \emph{local} problem (in the sense that we could fix the regularity constraints for each node $v$ by adding some dummy nodes only concerned with $v$), whereas regularizing 3-uniform hypergraphs appears to be a more \emph{global} problem (forcing us to simultaneously regularize all vertex-pairs while not introducing unintended $k$-cliques).

Due to this obstacle, we adopt a very different and conceptually cleaner approach. Let $G$ be a given $k$-partite $h$-uniform hypergraph. The new idea is to construct a regular $k$-partite $h$-uniform hypergraph~$G'$ from a graph product operation that we call the \emph{signed product} of~$G$ with a suitable \emph{signed} hypergraph $T$. A signed hypergraph $T$ has \emph{positive} and \emph{negative} edges $E^+(T)$ and $E^-(T)$, respectively. We define the resulting graph~$G'$ as follows: Its vertex set is $V'= V(G)\times V(T)$, and $(u_1, x_1), \dots, (u_h, x_h)\in V'$ form a hyperedge in $G'$ if and only~if
\smallskip
\begin{enumerate}
\item $\{u_1,\dots,u_h\}\in E(G)$ and $\{x_1,\dots, x_h\}\in E^+(T)$, or
\item $\{u_1,\dots, u_h\} \notin E(G)$ and $\{x_1,\dots, x_h\}\in E^-(T)$.
\end{enumerate}
\smallskip
We prove (Lemma~\ref{lemma:bichromatic-product}) that $G'$ yields an equivalent, but \emph{regular} instance for $k$-clique if the following conditions hold:
\smallskip
\begin{enumerate}[(i)]
    \item There is some $r \in \mathbb{N}$ such that for all nodes $x_1, \dots, x_{h-1} \in V(T)$ (taken from different parts in the $k$-partition of $T$), the set $\{x_1,\dots, x_{h-1}\}$ is contained in exactly $r$ positive and in exactly $r$ negative edges of $T$, 
    \item there is a $k$-hyperclique in $T$ involving only positive edges, and
    \item there is no $k$-hyperclique in $T$ involving a negative edge.
\end{enumerate}
\smallskip
We call a signed hypergraph satisfying these conditions a \emph{template graph}.
Our task is thus simplified to finding suitable template graphs, which \emph{only depend on $k$ and $h$, not on $n$}.

In a second step, we then show how to construct such template graphs in Lemma~\ref{lemma:template-graph-construction}. We employ a construction akin to Cayley graphs for some appropriate group $\mathcal G$ (concretely, we will pick $\mathcal G = (\mathbb{Z} / h \mathbb{Z})^d$ for some appropriate dimension $d = d(k, h)$). That is, let the nodes of $T$ be $k$ independent copies of $\mathcal G$ denoted $T_1, \dots, T_k$. For each $h$-tuple of vertices, say, $x_1 \in T_1, \dots, x_h \in T_h$, we let $\{x_1, \dots, x_h\}$ be a positive edge if and only if $x_1 + \dots + x_h = a^+$ for some fixed group element $a^+ \in \mathcal G$, and we let $\{x_1, \dots, x_h\}$ be a negative edge if and only if~\makebox{$x_1 + \dots + x_h = a^-$} for some other fixed group element $a^- \in \mathcal G$. The advantage of this construction is that it is immediately clear that the resulting graph~$T$ satisfies the regularity condition (i): For each $x_1 \in T_1, \dots, x_{h-1} \in T_{h-1}$ there is a unique choice~\makebox{$x_h := a^+ - (x_1 + \dots + x_{h-1}) \in T_h$} such that $\{x_1, \dots, x_h\}$ is a positive hyperedge; the same applies to negative edges and all other combinations of parts $T_{i_1}, \dots, T_{i_h}$. We will then pick the elements $a^+$ and $a^-$ in a specific way to guarantee the remaining conditions~(ii) and~(iii). 

\medskip
Below, in Sections~\ref{sec:intro-optCSPs} and~\ref{sec:intro-outlook}, we discuss applications of Theorem~\ref{thm:main1}.

\subsection{Optimizing Boolean CSPs over Weight-\texorpdfstring{\boldmath$k$}{k} Assignments}
\label{sec:intro-optCSPs}
The complexity of constraint satisfaction problems (CSPs) has been intensively studied. Starting with Schaefer's dichotomy theorem for Boolean CSPs~\cite{Schaefer78}, which initiated the eventually successful quest of proving the Dichotomy Conjecture~\cite{Bulatov17,Zhuk17,Zhuk20}, classification results for a number of settings have been proven: This includes, among others, approximability characterizations of natural optimization variants~\cite{Creignou95, KhannaSTW00}, as well as parameterized complexity classifications of solving CSPs parameterized by their \emph{solution size} (aka \emph{weight}) $k$, i.e., the number of non-zeros, see~\cite{Marx05, BulatovM14, KratschMW16}. The parameterized classification of Boolean CSPs due to Marx~\cite{Marx05} has subsequently been refined to a classification of the fine-grained time complexity~\makebox{$f(k)\cdot n^{g(k)}$}, resulting in 4 different regimes~\cite{KunnemannM20}: an FPT regime, a subexponential regime, a clique regime, and a brute-force regime.

In this work, we investigate the fine-grained complexity of \emph{optimizing} Boolean CSPs parameterized by solution size, i.e., optimizing the number of satisfied constraints among assignments of weight $k$. Formally, let $\Fam$ be a finite family of Boolean constraint functions. The problem $\maxcspk(\Fam)$ ($\mincspk(\Fam)$) asks to determine, given a set $\Phi$ of constraints -- each formed by applying some $f\in \Fam$ to a selection of the Boolean variables $x_1,\dots, x_n$ -- the maximum (minimum) number of constraints satisfied by a Boolean assignment that sets precisely $k$ variables to true (a \emph{weight-$k$} assignment). Among others, this class of problems contains the well-studied Densest $k$-Subgraph and Densest $k$-Subhypergraph problem ($\Fam$ consists of the $h$-ary AND), the Partial Vertex Cover in $h$-uniform hypergraphs ($\Fam$ consists of the $h$-ary OR), and the graph problem MaxCut($k$) of maximizing the cut size among cuts with set sizes $k$ and $n-k$, respectively ($\Fam = \{ \xor \}$), studied in e.g.~\cite{Cai08}.

Note that $\maxcspk(\Fam)$ trivially generalizes the previously studied variant $\cspk(\Fam)$, asking whether \emph{all} constraints are satisfiable by a weight-$k$ assignment, that has been fully classified in~\cite{Marx05, KunnemannM20}. Indeed, for some constraint families $\Fam$, the optimization variant $\maxcspk(\Fam)$ turns out to be harder than $\cspk(\Fam)$: $\cspk(\{\xor\})$ is in FPT~\cite{KunnemannM20}, but $\maxcspk(\{\xor\})$ is $W[1]$-hard~\cite{Cai08}. Thus, a classification result for $\maxcspk(\Fam)$ must differ from the classification of $\cspk(\Fam)$, but how significantly? Do we still obtain 4 regimes of tractability?

We shall prove that the answer is No: for $\maxcspk(\Fam)$ we obtain only 3 regimes: (1) a linear-time (rather than FPT) regime, (2) a clique regime, and (3) a brute-force regime. This characterization is governed by a different concept as well: rather than the notion of $\nand_h$-representing and $\impl$-representing families (see~\cite{KunnemannM20}), it is crucial to analyze the maximum \emph{degree} $\deg(f)$ over $f\in \Fam$, where $\deg(f)$ is defined as the degree of the (unique) multilinear polynomial representing $f$ (see Section~\ref{sec:csp-main}).

\begin{theorem}[Informal Version]\label{thm:main2}
Let $\mathcal{F}$ be constraint family. Let $d \coloneqq \max_{f\in \mathcal{F}} \deg(f)$. Then,
\begin{itemize}
   \item  If $d\le 1$, $\maxcspk(\Fam)$ and $\mincspk(\Fam)$ are linear-time solvable,
   \item if $d = 2$, the time complexity for $\maxcspk(\Fam)$ and $\mincspk(\Fam)$ is $f(k) \cdot n^{\omega (k/3)\pm O(1)}$, conditioned on the $k$-Clique Hypothesis,\footnote{The precise bound is~\smash{$f(k) \cdot n^{\omega(\lfloor k/3 \rfloor, \lceil k/3 \rceil, \lceil (k-1)/3\rceil)\pm o(1)}$} (matching the Ne\v{s}et\v{r}il-Poljak running time for $k$-Clique).}
   \item if $d \ge 3$, the time complexity for $\maxcspk(\Fam)$ and $\mincspk(\Fam)$ is $f(k) \cdot n^{k\pm o(1)}$, conditioned on the $k$-Hyperclique Hypothesis.
\end{itemize}
\end{theorem}

We  briefly mention some corollaries: Let us call a constraint function \emph{non-reducible} if it depends on all of its inputs. We can without loss of generality assume that all $f\in \Fam$ are non-reducible (since each $f\in \Fam$ can be replaced by an equivalent non-reducible function). As discussed, e.g., in Williams~\cite[Section 6.5]{Wil07}, non-reducible functions of degree at most~$2$ must have arity at most 4. Thus, there is only a \emph{finite} list of non-reducible function of degree at most 2, and thus assuming the 3-uniform hyperclique hypothesis, there exists only a finite number of constraint families for which $\maxcspk(\Fam)$ admits algorithms beating brute-force running time $f(k)n^{k\pm o(1)}$. 
Some interesting degree-$2$ functions (beyond binary arity) include:
\begin{itemize}
    \item \textsf{3-NAE} (as well as \textsf{3-AllEqual})
    \item \textsf{Sort($x_1,x_2,x_3,x_4$)}, i.e., the predicate that is true if and only if the bit string $x_1x_2x_3x_4$ is sorted descending or ascending.
    \item \textsf{Select}($x_1,x_2,x_3$) on 3 variables, i.e., the predicate \emph{if $x_1$ then $x_2$ else $x_3$}
\end{itemize}

We remark that a conceptually similar classification result has been proven for optimizing first-order graph formulas~\cite{BringmannCFK22}.

\subparagraph*{Technical Remarks.}
The algorithmic part of the above theorem is a straightforward adaption of a similar reduction of Weighted Degree-2 CSP to $k$-clique detection given in~\cite[Section 6.5]{Wil07}, see also~\cite{LincolnWW18, BringmannCFK22}. Our main technical contribution for Theorem~\ref{thm:main2} is to prove that this algorithmic technique is essentially the best possible, unless the 3-uniform hyperclique hypothesis can be refuted. 

For ease of presentation, consider first a constraint family $\Fam$ containing a constraint function $f:\{0,1\}^3\to\{0,1\}$ that is \emph{symmetric} and has $\deg(f) = 3$. Let $\alpha,\beta,\gamma,\delta$ be such that $f(x,y,z) = \alpha \cdot xyz + \beta \cdot (xy + yz + xz) + \gamma \cdot (x+y+z) + \delta$ and note that $\alpha \ne 0$. The main idea is the following: Given a $k$-partite regular 3-uniform hypergraph $G=(V_1 \cup \cdots \cup V_k, E)$ with $|V_i| = n$ for $i \in [k]$, we create a Boolean CSP by introducing, for each $\{a,b,c\} \in E(G)$, the constraint $f(x_a,x_b,x_c)$. Since $G$ is regular, there is some $\lambda$ such that any pair of nodes not from the same part is incident to precisely $\lambda$ edges. Thus, a weight-$k$ assignment respecting the $k$-partition, i.e., the 1-variables $x_{i_1},\dots, x_{i_j}$ satisfy $i_j \in V_j$, has an objective value of

\[ \alpha \cdot |\{i_a, i_b, i_c\}\in E|+\beta \cdot \binom{k}{2} \lambda + \gamma \cdot k(k-1)n\lambda + \delta \cdot |E|.\]
The $\alpha$-term is the only term depending on the choice of the assignment, maximizing the objective (for $\alpha > 0$, $\alpha < 0$ can be handled similarly) if and only if the assignment corresponds to a clique in $G$. 
For $\beta > 0$, $k$-partite assignments are favorable over non-partite assignments. For $\beta \le 0$, we instead need to make sure that $k$-partite assignments are favorable, by carefully introducing dummy constraints. 
Finally, we show how to handle $\maxcspk(\Fam)$ if $\Fam$ contains any (not necessarily symmetric) constraint function with degree at least 3.

\subsection{Outlook and Further Applications}
\label{sec:intro-outlook}
We believe that reducing from the regular Hyperclique problem will find further applications. One interesting theme is that by plugging our regularization into known reductions, we can in many cases obtain conditional lower bounds for regular instances. For instance, conditioned on the $3$-Uniform $4$-Hyperclique Hypothesis, Dalirrooyfard and Vassilevska Williams~\cite{DalirrooyfardW22} proved that deciding if a graph contains an induced 4-cycle requires time $n^{2-o(1)}$ even in graphs with $\bigO(n^{3/2})$ edges. Composing this reduction with our regularization reduction (in a white-box manner), one can conclude that induced 4-cycle detection takes time $n^{2-o(1)}$ even in \emph{regular} graphs with $\bigO(n^{3/2})$ edges.

\subparagraph*{Regularization for Non-\boldmath$k$-Partite Hypergraphs?}
An interesting follow-up question to our work is whether we can obtain a similar regularization reduction for the non-$k$-partite Hyperclique problem. This question is of lesser importance for the design of reductions as in most (though not all) contexts it is more convenient to work with the $k$-partite version. Interestingly, already constructing an interesting NO instance to the problem -- i.e., a 3-uniform hypergraph $G = (V, E)$ that is dense (say, has at least $|E| \geq n^{3-o(1)}$ hyperedges), regular (i.e., each pair of distinct vertices $v_1, v_2 \in V$ appears in the same number of hyperedges) and does not contain a 4-clique -- appears very challenging. Even disregarding the 4-clique constraint, constructing (non-$k$-partite) regular hypergraphs falls into the domain of \emph{combinatorial block designs} and is known to be notoriously difficult.

\section{Preliminaries}\label{sec:preliminaries}
We denote by $[n] = \{1,\dots,n\}$ for $n \in \mathbb N$.
For an $n$-element $S$ and $0 \leq k \leq n$, we denote by $\binom{S}{k}$ the set of all size $k$ subsets $L \subseteq S$.
Let $S_n$ denote the set of all permutations of a set on $n$ elements.
For a hypergraph $G = (V,E)$, let $|G| = |V|+|E|$. 
A simple graph is $r$-regular, if every vertex is incident to exactly $r$ edges.
This notion can be generalized in the following way:
For $h \in \mathbb N$ and $s < h$, an $h$-uniform hypergraph is \emph{$(s, \reg)$-regular} if every possible $s$ tuple of pairwise distinct vertices is contained in exactly $\reg$ edges. 
Using our notation, the notion of $r$-regularity of a graph can thus be viewed as $(1,r)$-regularity of a $2$-uniform hypergraph.

We will mostly work with balanced, $k$-partite $h$-uniform hypergraphs $G = (V_1\cup\dots \cup V_k, E)$ where
1) $V_i$ are pairwise disjoint and $|V_i| = |V_j|$ for all $i \neq j \in [k]$
2) Each edge contains precisely $h$ distinct vertices and
3) No edge contains two vertices $v,v'$ such that $v,v'\in V_i$ for any $i\in [k]$.

For the purposes of this paper, we will relax the notion of regularity to better fit the $k$-partite setting. 
In particular, by slight abuse of notation, we say that a $k$-partite $h$-uniform hypergraph $G$ is \emph{$(s, \reg)$-regular} if for each tuple of pairwise distinct indices $i_1,\dots, i_s\in [k]$ and for any choice of vertices $v_1\in V_{i_1},\dots, v_s\in V_{i_s}$, there are precisely $\reg$ many edges $e\in E$ such that $\{v_1,\dots, v_s\} \subseteq e$.
We call an $h$-uniform hypergraph $G$ 
\emph{$\reg$-regular} if it is $(h-1,\reg)$-regular.

We state the following simple observation for $(s, \reg)$-regular hypergraphs.
\begin{observation}\label{obs:reguarity-of-smaller-arities}
    For $s<h<k$, let $G$ be a $(s, \reg)$-regular $k$-partite $h$-uniform hypergraph. Then for any $s'\leq s$ there exists $\reg'$, such that $G$ is also $(s', \reg')$-regular.
\end{observation}

\subsection{Hardness Assumptions}
The \emph{$k$-clique} problem asks, given a graph $G = (V,E)$ with $|V| = n$, to decide if $G$ contains a clique (set of pairwise adjacent vertices) of size $k$. For $k$ divisible by $3$, 
$k$-clique detection can be solved in time $\bigO(n^{\omega k/3})$~\cite{NesetrilP85} where $\omega < 2.372$~\cite{AlmanDWXXZ25} is the matrix multiplication exponent. 
Generalizing $k$-clique detection to $h$-uniform hypergraphs $G=(V,E)$ asks to determine if there exists a $k$-(hyper)clique $C \subset V$ of size $k$ such that $\binom{C}{h} \subseteq E$. In contrast to the $2$-uniform case, the matrix multiplication approach fails to generalize to $h\geq 3$ (see~\cite{LincolnWW18} for a discussion) and no known algorithm is significantly faster than $\bigO(n^k)$. 
Improvements for $k$-hyperclique would entail faster algorithms for problems that are believed to be hard, such as \textsc{Max-$h$-SAT}. 
This gives raise to the following conjecture:
\begin{hypothesis}[$h$-Uniform $k$-Hyperclique Hypothesis]
    Let $\epsilon > 0$ and $k > h$.
    \begin{enumerate}
        \item For $h = 2$ there is no $\bigO(n^{({\omega k}/{3}) - \varepsilon})$-time algorithm detecting a $k$-clique in a $k$-partite graph.
        \item For $h\geq 3$, there is no $\bigO(n^{k - \varepsilon})$-time algorithm detecting a $k$-clique in a $k$-partite $h$-uniform hypergraph.
    \end{enumerate}
\end{hypothesis}

\section{Making Hypergraphs Regular} \label{sec:regularizing-hypergraphs}
In this section we break down in detail the approach to regularize $h$-uniform hypergraphs, while preserving large cliques.
More precisely, we prove the following main theorem.
\begin{theorem}\label{th:hc-regularity-construction}
    Let $h < k$ be constants and let $G$ be an $h$-uniform hypergraph. Then there exists a $k$-partite $h$-uniform hypergraph $G'$ satisfying the following conditions:
    \begin{enumerate}
        \item For any $1\leq s<h$, there exists some $\reg$ such that $G'$ is $(s,\reg)$-regular.
        \item $G'$ contains a clique of size $k$ if and only if $G$ contains a clique of size $k$.
        \item There exists a computable function $f$, such that $|V(G')| = f(k)\cdot |V(G)|$.
        \item $G'$ can be computed deterministically in time $\bigO(|G'|)$.
    \end{enumerate}
\end{theorem}
Moreover, as a direct consequence of Theorem \ref{th:hc-regularity-construction}, we get the following theorem for free.
\begin{theorem}\label{th:regular-hc-hypothesis}
    Let $h \geq 2$. The $k$-Clique Detection problem is solvable in time $f(k) \cdot n^{g(k)}$ on $k$-partite $h$-uniform hypergraphs if and only if it is solvable in time $f'(k) \cdot n^{g(k)}$ on \emph{regular} $k$-partite $h$-uniform hypergraphs (for some computable functions $f, f'$ and $h\leq g(k)\leq k$).
\end{theorem}
In order to prove Theorem \ref{th:hc-regularity-construction} we come up with an appropriate notion of hypergraph product, that we call \emph{signed product}. It takes an arbitrary hypergraph $G$ and regular hypergraph $T(h,k)$ equipped with a sign function and certain properties (we will call such a hypergraph \emph{template hypergraph}) and produces the hypergraph $G'$, satisfying the conditions of the theorem. 
The proof of Theorem \ref{th:hc-regularity-construction} thus consists of two main parts:
\begin{enumerate}
    \item Defining the proper notion of the \emph{signed product} that produces the hypergraph $G'$ satisfying the conditions of the theorem.
    \item Showing that for any choice of $2\leq h<k$ we can construct a template hypergraph $T(h,k)$.
\end{enumerate}
\medskip
We start by formally defining the notion of a \emph{signed hypergraph}.
\begin{definition}[Signed Hypergraph]
    Let $h<k$ be constants. A \emph{signed hypergraph} is a $k$-partite, $h$-uniform hypergraph $T = (T_1 \cup \dots\cup T_k, E, \sigma)$ equipped with a sign function $\sigma: E\to \{1,-1\}$. 
\end{definition}
We say an edge $e$ is \emph{positive} if $\sigma(e)=1$ and denote by $E^{+} := \{e\in E\mid \sigma(e) = 1\}$ the set of all positive edges. We define  the set of \emph{negative} edges $E^-$ analogously.
Using this notation, we equivalently represent a signed hypergraph as a tuple $T = (T_1\cup\dots \cup T_k, E^+,E^-)$.

With the concept of signed hypergraphs set, we are ready to define the \emph{signed product}.
\begin{definition}[Signed Product]
    Given an $h$-uniform hypergraph $G=(V,E)$ and a signed hypergraph $T = (T_1\cup\dots\cup T_k, E^+,E^-)$, the \emph{signed product} of $G$ and $T$ is a $k$-partite, $h$-uniform hypergraph $G'$ defined as follows:
    \begin{itemize}
        \item $V(G') = V\times (T_1\cup \dots \cup T_k)$.
        \item Let $\{t_1,\dots,t_h\} \in E^+$. Then for all $u_1,\dots, u_h\in V$ that form an edge in $G$, let $\{(u_1,t_1),\dots, (u_h,t_h)\}$ be an edge in $G'$.
        \item Let $\{t_1,\dots,t_h\} \in E^-$. Then for all, not necessarily distinct, $u_1,\dots, u_h\in V$ that form a non-edge in $G$, let $\{(u_1,t_1),\dots, (u_h,t_h)\}$ form an edge in $G'$.

    \end{itemize}
\end{definition}
Given a hypergraph $G$, in order for our product graph $G'$ to satisfy the properties of {Theorem~\ref{th:hc-regularity-construction}} it is not sufficient to take just any signed hypergraph $T$, but we want our signed hypergraph to have a certain structure. 
To this end, we introduce the notion of \emph{template hypergraphs}.
\begin{definition}[Template Hypergraph]
    For fixed constants $h < k$, we call $T(h,k) = (T_1\cup \dots\cup T_k, E^+,E^-)$ a template hypergraph, if the following properties are satisfied:
    \begin{enumerate}
        \item\label{prop:template:1} There exists a positive integer $\reg$, such that the underlying ``monochromatic'' hypergraphs $T^+ := (T_1\cup\dots\cup T_k, E^+)$ and $T^-:=(T_1\cup\dots\cup T_k, E^-)$ are \emph{both} $(h-1, \reg)$-regular.
        \item\label{prop:template:2} The hypergraph $T^+$ contains a clique of size $k$.
        \item\label{prop:template:3} No clique of size $h+1$ contains edges from $E^-$.
    \end{enumerate}
\end{definition}
\smallskip
Equipped with the definitions above, we are now ready to formally state the main lemma of this section.
\begin{lemma}\label{lemma:bichromatic-product}
    For any fixed constants $k,h$ with $h<k$, let $G = (V,E')$ be an $h$-uniform hypergraph and $T(h,k) = (T_1\cup \dots \cup T_k, E,\sigma)$ be a template hypergraph. Let $G'$ be the signed product of $G$ and $T(h,k)$. Then $G'$ satisfies all the conditions of Theorem \ref{th:hc-regularity-construction}.
\end{lemma}
Before proving this lemma however, let us first prove some auxiliary statements.
For the rest of this section, let $G$ denote an arbitrary $h$-uniform hypergraph, $T(h,k) = (T_1 \cup\dots,T_k,E^+,E^-)$ a template hypergraph and $G'$ the signed product of $G$ and $T(h,k)$.
\begin{restatable}{lemma}{lemmaCliqueConservation}\label{lemma:cliques-in-bichromatic-product}
    Let $(u_1,t_1), \dots, (u_k,t_k)$ be $k$ arbitrary pairwise distinct vertices in $G'$. Let $\pi:[k]\to [k]$ be an arbitrary permutation. The following statements are equivalent:
    \begin{enumerate}
      \item The vertices $u_{1}, \dots, u_k$ form a clique in $G$ \emph{and} $t_1, \dots, t_k$ form a clique in $T(h,k)$.
        \item The vertices $(u_{\pi(1)},t_1), \dots, (u_{\pi(k)},t_k)$ form a clique in $G'$.
    \end{enumerate}
\end{restatable}
\medskip
\begin{proof}
    Assume first that the vertices $u_1, \dots, u_k$ form a clique in $G$ and the vertices $t_1, \dots, t_k$ form a clique in $T(h,k)$.
    Since $T(h,k)$ is a template hypergraph (and $k\geq h+1$), for all indices $i_1,\dots, i_h\in [k]$, it holds that $\{t_{i_1}, \dots t_{i_h}\}\in E^+$ . 
    Hence, by definition of the signed product, and since $u_1,\dots, u_k$ form a clique in $G$, it follows that $(u_{\pi(1)},t_1), \dots, (u_{\pi(k)},t_k)$ form a clique in $G'$.

    Assume now that the vertices $(u_{\pi(1)},t_1), \dots, (u_{\pi(k)},t_k)$ form a clique in $G'$.
    Consider (without loss of generality) the edge $\{(u_{\pi(1)},t_1), \dots, (u_{\pi(r)},t_h)\}$. 
    By definition of the signed product, it holds that either 
    \begin{enumerate}
        \item $\{u_{\pi(1)}, \dots, u_{\pi(r)}\}$ is an edge in $G$, \emph{and} $\{t_1, \dots, t_h\}\in E^+$, or
        \item $\{u_{\pi(1)}, \dots, u_{\pi(r)}\}$ is a non-edge in $G$, \emph{and} $\{t_1, \dots, t_h\} \in E^-$.\label{item:2}
    \end{enumerate}
    Looking at these two conditions, it is clear that the set $\{t_1,\dots, t_k\}$ forms a clique in $T(h,k)$, and moreover, since $T(h,k)$ is a template hypergraph, it must hold that all of the edges in this clique are positive. 
    Hence, condition \ref{item:2} never occurs and it follows that also $u_1,\dots, u_k$ form a clique in $G$, as desired.
\end{proof}

For a hypergraph $X$, let $C_k(X)$ denote the number of cliques of size $k$ in~$X$. 
Following the approach as in the proof of the previous lemma, we obtain this corollary.
\begin{corollary}
    Let $T(h,k), G$ and $G'$ be as above. Then 
    $C_k(G') = k!\cdot C_k(G)\cdot C_k(T(h,k))$.
\end{corollary}
\begin{restatable}{lemma}{lemmaBichromaticProductRegularity}\label{lemma:regularity-of-bichromatic-product}
    Let $T(h,k), G$ and $G'$ be as above. Let $\reg$ be such that both underlying hypergraphs $T^+$ and $T^-$ are $\reg$-regular. Then $G'$ is a $k$-partite $\reg\cdot|V(G)|$-regular hypergraph.
\end{restatable}
\begin{proof}
    Take any set of $h-1$ vertices $\{(u_1,t_1),\dots, (u_{r-1}, t_{h-1})\}$ such that $t_1\in T_{i_1},\dots, t_{h-1}\in T_{i_{h-1}}$ for pairwise distinct $i_1,\dots, i_{h-1}$.
    Assume there are $\alpha$ many vertices $u_h$ in $G$ such that $\{u_1,\dots, u_{h-1}, u_h\}\in E(G)$.
    Since there are $\reg$ many vertices $t_h$, such that $\{t_1,\dots, t_{h}\}\in E^+$, by the definition of the signed product, for each such pair $(u_h, t_h)$, the set  $\{(u_1,t_1),\dots, (u_{r}, t_{h})\}$ forms an edge in $G$. In fact these are all of the edges that contain $\{(u_1,t_1),\dots, (u_{r-1}, t_{r-1})\}$ and that stem from the positive edges in $T(h,k)$.
    On the other hand, there are clearly $|V(H)| - \alpha$ many vertices $u_h$ in $G$ such that $\{u_1,\dots, u_{r-1}, u_h\}$ forms a non-edge in $G$. 
    Since there are also $\reg$ many vertices $t_h$, such that $\{t_1,\dots, t_{h}\}\in E^-$, by the definition of the signed product, for each such pair $(u_h, t_h)$, the set  $\{(u_1,t_1),\dots, (u_{h}, t_{h})\}$ forms an edge in $G$. Similarly as above, these are all of the edges that contain $\{(u_1,t_1),\dots, (u_{h-1}, t_{h-1})\}$ and that stem from the negative edges in $T(h,k)$.
    
    Note that the set of edges in $G$ that stems from the positive edges in $T(h,k)$ and the set of edges in $G$ that stems from the negative edges in $T(h,k)$ are disjoint, thus the total number of edges $e\in E(G)$ such that $\{(u_1,t_1),\dots, (u_{h-1}, t_{h-1})\}\subset e$ can be computed by summing those two values.
    Hence, we obtain the number of such edges as $\reg\alpha + \reg(|V(G)|-\alpha) = \reg\cdot(|V(G)|)$. In particular this number does not depend on $\alpha$ and $G'$ is $\reg\cdot(|V(G)|)$-regular.
\end{proof}

\subsection{Constructing Template Hypergraphs} \label{sec:template-graphs}
In this section we show that for any choice of constants $h<k$ we can efficiently construct a template hypergraph $T(h,k)$, which is the missing ingredient in the proof of Theorem \ref{th:hc-regularity-construction}. In particular, we prove the following lemma
\begin{lemma}\label{lemma:template-graph-construction}
    For any fixed constants $h<k$, there exists a template hypergraph $T(h,k)$.
\end{lemma}
This template hypergraph can be deterministically constructed in time $\bigO(f(h,k))$ for some computable function $f$. 

We first give the construction, then prove that it fulfills the properties of a template hypergraph.
For the rest of this section let $h < k$ be constants. Consider the additive group $\G := \left(\mathbb Z/h\mathbb Z\right)^{\binom{k}{h}}$, i.e. the group of all vectors of length $\binom{k}{h}$ over the group $\mathbb Z/h\mathbb Z$. We index the dimensions of these vectors by subsets of $[k]$ of size $h$.
For any set $S\in \binom{[k]}{h}$, denote by $e_S$ the vector in $\G$ that has a $1$ in the dimension indexed by $S$ and zeros in all other dimensions, and by $\overline 0$ we denote the additive identity in $\G$ (the all-zero vector). 

Consider the $k$-partite signed graph $T(h,k) = (T_1\cup\dots \cup T_k, E^+,E^-)$ where each $T_i$ corresponds to a copy of the group $\G$ and the edges are added as follows.
For each $S:=\{i_1,\dots, i_h\}\in \binom{[k]}{h}$ and each choice of vertices $v_1\in T_{i_1},\dots, v_h\in T_{i_h}$, we add a positive edge (i.e. $\{v_1,\dots, v_h\}\in E^+$) if the corresponding elements of $\G$ $x_1,\dots, x_h$ satisfy the equality $x_1+\dots + x_h = \overline 0$.
Moreover, if the corresponding elements satisfy $x_1+\dots + x_h = e_S$, we add a negative edge (i.e. $\{v_1,\dots, v_h\}\in E^-$).
We proceed to prove that the constructed $T(h,k)$ has the desired properties. We begin by observing that the number of vertices only depends on $k$ and $h$.
\begin{observation}\label{obs:template-graph-clique}
    The signed hypergraph $T(h,k)$ has $kh^{\binom k h}$ vertices.
\end{observation}
We now state the regularity conditions of the underlying hypergraphs $T^+$ and $T^-$.
\begin{restatable}[P. \ref{prop:template:1}]{lemma}{LemmaRegularityOfTemplate}\label{lemma:lambda-regularity-of-T}
    Let $T(h,k)$ be a signed graph as constructed above and let $\reg := k-h+1$. Then the underlying hypergraphs $T^+ := (T_1\cup\dots\cup T_k, E^+)$ and $T^-:=(T_1\cup\dots\cup T_k, E^-)$ are both $\reg$-regular.
\end{restatable}
\begin{proof}
    Let $S:=\{i_1,\dots, i_{h-1}\}\in \binom{[k]}{h-1}$ be arbitrary. Let $v_1\in T_{i_1},\dots, v_{h-1}\in T_{i_{h-1}}$ be any choice of $h-1$ vertices that correspond to the elements $x_1,\dots, x_{h-1}$. 
    Take any $i_h\in [k]\setminus S$. There is a unique element $x_h:=-(x_1+\dots + x_{h-1}) \in \G$ that satisfies $x_1+\dots+x_h = \overline{0}$, and hence the vertex $v_h\in T_{i_h}$ that corresponds to this element is the unique vertex in $T_{i_h}$ such that $\{v_1,\dots, v_h\}\in E^+$.
    Iterating over all the choices in $i_h\in [k]\setminus S$ we can conclude that there are precisely $k-h+1 = \reg$ vertices $v_h\in V(T)$ such that $\{v_1,\dots, v_h\}\in E^+$.
    
    A very similar argument shows that there are also precisely $k-h+1 = \reg$ vertices $v_h\in V(T)$ such that $\{v_1,\dots, v_h\}\in E^-$.
\end{proof}

It remains to show that the hypergraph $T^+$ contains a clique of size $k$, and that no clique of size $h+1$ contains an edge from $E^-$.
We begin by showing the latter.
\begin{lemma}[{P. \ref{prop:template:3}}]\label{lemma:no-negative-edges}
    Each clique in $T(h,k)$ of size $h+1$ consists exclusively of positive edges.
\end{lemma}
\begin{proof}
    Let $\{i_1,\dots, i_{h+1}\}\in \binom{k}{h+1}$, and assume for contradiction that there exists a clique of size $h+1$ consisting of vertices $v_1\in T_{i_1},\dots, v_{h+1}\in T_{i_{h+1}}$, such that (w.l.o.g.) $\{v_1,\dots, v_h\}\in E^-$.
    Let $S:=\{i_1,\dots, i_h\}$ and $x_i$ be the elements that corresponds to $v_i$. 
    By construction of $T(h,k)$, this means that $x_1,\dots, x_h$ satisfy the equation $x_1+\dots+x_h = e_S$. 
    Moreover, since $v_1,\dots, v_{h+1}$ form a clique, for any subset $S_j := \{i_1,\dots, i_{j-1}, i_{j+1}, \dots i_{h+1}\}$ with $S_j\neq S$, 
    we have
    \[
        x_1+\dots+ x_{j-1}+ x_{j+1} + \dots+ x_{h+1} \in \{\overline 0, e_{S_j}\}.
    \]
    As $S_j\neq S$, if we consider only the entry indexed by $S$, we get $x_1[S]+\dots+ x_{j-1}[S] + x_{j+1}[S] + \dots+ x_{h+1}[S] = 0$.
    Over all possible values of $j$, we obtain the following equation system.
    \begin{equation*}
        \begin{bmatrix} 0 & 1 & 1 & \cdots & 1\\ 1 & 0 & 1  & \cdots & 1 \\ 1 & 1 & 0 & \cdots & 1 \\ \vdots & \vdots & \vdots & \ddots \\ 1 & 1 & 1 & \cdots & 0\end{bmatrix} 
        \cdot 
        \begin{bmatrix} x_1[S] \\ \vdots \\ x_{h+1}[S] \end{bmatrix} = \begin{bmatrix} 0 \\ \vdots \\ 0 \\1 \end{bmatrix}.
    \end{equation*}
    Summing up all the rows, we get that
    \(
    h (x_1[S] +\dots + x_{h+1}[S]) = 1
    \).
    Recalling that $h=0$ in $\mathbb Z / h \mathbb Z$, we get that $0=1$, which is a contradiction.
\end{proof}
We can finally observe that choosing a vertex $v_i\in T_i$ corresponding to the zero element $\overline{0}\in \G$ for each $V_i$ yields a clique of size $k$ in $T(h,k)$. We state this as a separate observation.
\begin{observation}[{P. \ref{prop:template:2}}]
    Let $v_1\in T_1,\dots v_k\in T_k$ be the vertices corresponding to the all-zero vector in their respective parts. Then $\{v_1,\dots, v_k\}$ is a clique in the hypergraph $T^+$.
\end{observation}
This concludes the proof of Lemma \ref{lemma:template-graph-construction}.
We are also ready to prove Theorems \ref{th:hc-regularity-construction} and \ref{th:regular-hc-hypothesis}.
\begin{proof}[Proof of Theorem \ref{th:hc-regularity-construction}.]
    Given an arbitrary $h$-uniform hypergraph $G$, construct the signed hypergraph $T(h,k)$ as above. By Lemma \ref{lemma:template-graph-construction}, this is a template hypergraph. 
    Let $G'$ be the signed product of $G$ and $T(h,k)$. 
    Then, by Lemma \ref{lemma:bichromatic-product} $G'$ has all the desired properties.
\end{proof}
\begin{proof}[Proof of Theorem \ref{th:regular-hc-hypothesis}.]
    Assume that we can decide if an $(s,\reg)$-regular $h$-uniform hypergraph contains a clique of size $k$ in time $f(k)\bigO(n^{g(k)})$, for computable functions $f,g$ and $g(k)\geq h$.
    Then, given an arbitrary $h$-uniform hypergraph $G$ consisting of $n$ vertices, by Theorem \ref{th:hc-regularity-construction} we can construct an $(s,\reg)$-regular $h$-uniform hypergraph $G'$ in linear time, such that it contains a clique of size $k$ if and only if $G$ contains a clique of size $k$. Moreover, the number of vertices of $G'$ is bounded by $f'(k)n$ for some computable function $f'$. We then run the fast algorithm to detect cliques of size $k$ in time $\bigO((f'(k)n)^{g(k)}) = f''(k)n^{g(k)})$ and report this as the solution to the original instance.
    Note that the other direction is trivial.
\end{proof}
In fact, we can prove an even stronger result, showing that counting cliques of size $k$ in regular hypergraphs is as hard as counting cliques in general hypergraphs. 
\begin{theorem}\label{thm:counting}
    For any $h\geq 2$, there exists an algorithm counting the cliques of size $k$ in $k$-partite regular $h$-uniform hypergraphs in time $f(k)n^{g(k)}$ if and only if there exists an algorithm counting the cliques of size $k$ in general $h$-uniform hypergraphs in time $f'(k)n^{g(k)}$.
\end{theorem}
In order to prove this theorem, we extend Observation \ref{obs:template-graph-clique} to show that there is in fact a fixed number of cliques in our constructed template graph, only depending on $h$ and $k$.
Recall that by $\G$ we denote the group $\left(\mathbb Z/h\mathbb Z\right)^{\binom{k}{h}}$. 
\begin{restatable}{lemma}{LemmaCountingCliquesInTemplateGraph}
    Let $T(h,k)$ be a template hypergraph as constructed above. The number of cliques of size $k$ in $T(h,k)$ is precisely $|\G| = h^{\binom{k}{h}}$.
\end{restatable}
\begin{proof}
    Fix any element $x\in \G$ and for each $i\in [k]$ take the unique vertex $v_i$ in $T_i$ that corresponds to $x$. We claim that $v_1,\dots, v_k$ form a clique in $T(h,k)$.
    Indeed, if we take any pairwise distinct indices $i_1,\dots, i_h\in [k]$, the vertices $v_{i_1},\dots, v_{i_h}$ form a positive edge in $T(h,k)$, since $rx = 0$ holds for any $x\in \G$.
    Hence the number of cliques of size $k$ is at least $|\G|$. It remains to show the upper bound as well.

    Assume that $v_1\in T_1,\dots v_k\in T_k$ form a clique of size $k$ in $T(h,k)$. By Lemma \ref{lemma:no-negative-edges}, this clique contains no negative edges. In particular, each selection of $h$ pairwise distinct vertices forms a positive edge, which means that the elements $x_1,\dots, x_k\in \G$ corresponding to the vertices $v_1,\dots, v_k$ respectively satisfy the condition that $x_{i_1}+\dots+ x_{i_h} = \overline 0$ for any choice of pairwise distinct indices $i_1,\dots, i_h$.
    Consider the elements $x_1,\dots, x_h$. By the arguments above, we have that $x_1+\dots+x_h = \overline 0$.
    Now if we replace $x_h$ with $x_{h+1}$, we similarly have $x_1+\dots+x_{h+1} = \overline 0$.
    Subtracting these two equalities yields that $x_h-x_{h+1}= \overline 0$ and in particular this implies that $x_h=x_{h+1}$. 
    Recycling this argument we obtain that $x_1=\dots = x_{k}$, and thus up to the choice of an element of $\G$, there is a unique clique of size $k$ in the template hypergraph $T(h,k)$. 
\end{proof}

The proof of Theorem \ref{thm:counting} now follows directly.

\section{Boolean Constraint Optimization Problems}
\label{sec:csp-main}
In this section we present the main application of our construction, 
namely exploiting the regularity of hyperclique detection to show hardness of certain families of Boolean Constraint Optimization Problems. 
In particular, we obtain a full classification of the fine-grained complexity of this large class of problems.
The problems are defined as follows:
\begin{definition}[Boolean Constraint Optimization Problems ($\maxcspk/\mincspk$)]
    Let $\Fam$ be a finite Boolean constraint family (i.e. a set of functions $\phi:\{0,1\}^r \to \{0,1\}$). 
    Given a set $\Phi$ of $m$ Boolean constraints $C$ on variables 
    $x_1,\dots,x_n$, each of the form $\phi(x_{i_1},\dots, x_{i_r})$, where $\phi\in \Fam$,
    the problem $\maxcspk(\Fam)$ (resp. $\mincspk(\Fam)$) asks for an assignment $a:\{x_1,\dots, x_n\}\to \{0,1\}$ setting precisely $k$ variables to $1$ that maximizes (resp. minimizes) the number of satisfied constraints in $\Phi$. 
\end{definition}
It is well known that each Boolean constraint function corresponds to a unique multilinear polynomial (see e.g. \cite{Wil07}). More precisely, for any Boolean constraint function $\phi:\{0,1\}^r \to \{0,1\}$, there exists a unique multilinear polynomial $f$ such that $f(x_1,\dots, x_r) = \phi(x_1,\dots, x_r)$ for any $x_1,\dots, x_r\in \{0,1\}$.\footnote{E.g. ${\rm OR}_3(x,y,z)$ corresponds to the multilinear polynomial $f(x,y,z)=x+y+z-xy-xz-yz+xyz$}
We call such a polynomial the \emph{characteristic polynomial of $\phi$}.\footnote{This characterization is faithful in the sense that if two Boolean functions $\phi,\phi'$ have the same characteristic polynomial, they represent the same Boolean function.}
This characterization allows the concepts defined on multilinear polynomials to transfer to the domain of Boolean constraint functions.
For a Boolean constraint function $\phi$, we say that $\phi$ has \emph{degree} $d$ if its characteristic polynomial has degree $d$.
Similarly, the \emph{degree} of Boolean constraint family $\Fam$ is the smallest integer $d$ such that every constraint $\phi\in \Fam$ has degree at most $d$.

We are now ready to state the main theorem that we prove in this section:
\begin{theorem}[Hardness of the Boolean Constraint Optimization Problems]\label{th:CSP-hardness}
    Let $\Fam$ be a Boolean constraint family of degree $d$.
    \begin{itemize}
        \item If $d = 2$, then for no $\varepsilon>0$ does there exist an algorithm solving $\maxcspk(\Fam)$ (resp. $\mincspk(\Fam)$) in time $\bigO(n^{k(\omega/3 -\varepsilon)})$, unless the $k$-Clique Hypothesis fails.
        \item If $d \geq 3$, then for no $\varepsilon>0$ does there exist an algorithm solving $\maxcspk(\Fam)$ (resp. $\mincspk(\Fam)$) in time $\bigO(n^{k(1 -\varepsilon)})$, unless the $3$-uniform $k$-Hyperclique Hypothesis fails. \label{item:CSP-hardness-2}
    \end{itemize}
\end{theorem}
\medskip
We will show in detail the second reduction for families with degree at least $3$, reducing from $3$-uniform $k$-Hyperclique detection, 
the reduction from $k$-Clique is achieved similarly.
By simple modifications to the existing algorithms in the literature (see e.g. \cite{LincolnWW18,Wil07}), we can complement these hardness results with the corresponding algorithms to obtain a tight classification of the fine-grained complexity of Boolean Constraint Optimization Problems. The details can be found in the full paper.
\begin{theorem}[Boolean Constraint Optimization Algorithms]\label{thm:algorithms}
    Let $\Fam$ be an arbitrary constraint family and let $d$ be the degree of $\Fam$.
    \begin{itemize}
        \item If $d\leq 1$, we can solve $\maxcspk(\Fam)$ ($\mincspk(\Fam)$) in linear time $\bigO(m+n)$.
        \item If $d = 2$, we can solve $\maxcspk(\Fam)$ ($\mincspk(\Fam)$) in time $\bigO(n^{\omega (\lfloor k/3\rfloor,\lceil k/3\rceil,\lceil (k-1)/3\rceil)})$.
        \item If $d \geq 3$, we can solve $\maxcspk(\Fam)$ ($\mincspk(\Fam)$) in time $\bigO(n^{k})$.
    \end{itemize}
\end{theorem}
We prove all the hardness results and the algorithms for the $\maxcspk$ problem, and remark that there is a canonical reduction showing equivalence between $\maxcspk$ and $\mincspk$.
\subsection{Hardness Construction} This section is dedicated to proving Theorem \ref{th:CSP-hardness}.
But first, we will introduce some useful notation and prove some auxiliary results.
For the rest of this section, let $X$ be a set consisting of $n$ variable, let $\Fam$ be an arbitrary finite constraint family and let $\Phi$ be the given set consisting of $m$ constraints $C_1,\dots, C_m$ over $X$, where each constraint $C_i$ is of type $\phi_i(x_{i_1},\dots, x_{i_{r_i}})$ for some $\phi_i\in \Fam$ and some choice of $x_{i_1},\dots, x_{i_{r_i}}\in X$.
Let $a:X \to \{0,1\}$ be a function assigning to each variable a Boolean value. By $\Phi(a)$ we denote the number of constraints in $\Phi$ satisfied by $a$. 
The following observation plays a crucial role in our proof.
\begin{observation} \label{lemma:char-polynomial}
    For each constraint $C_i$, let $f_i$ denote the characteristic polynomial of the corresponding function $\phi_i$.
    For any assignment $a:X\to \{0,1\}$, the equality $\Phi(a) = \sum_{i\in m} f_i(a(x_{i_1}), \dots, a(x_{i_{r_i}}))$ holds.
\end{observation}

With this, it suffices to find an assignment that maximizes the value of the polynomial obtained by summing over all the characteristic polynomials of the functions corresponding to the constraints. 
More precisely, given a set $\Phi$ consisting of $m$ constraints $C_1,\dots, C_m$, define the polynomial $P_\Phi$ as 
\[
    P_\Phi(x_1,\dots, x_n) = \sum_{i=1}^m f_i(x_{i_1},\dots, x_{i_{r_i}})
\]
where $f_i$ denotes, as above, the characteristic polynomial of the function $\phi_i$, corresponding to the constraint $C_i$.
Using this notation, for any assignment $a$, as a simple consequence of Observation \ref{lemma:char-polynomial}, we have that $\Phi(a) = P_\Phi(a(x_1),\dots, a(x_n))$.

First, we prove the following lemma, simplifying our reduction, by allowing us to only focus on showing hardness for the families of degree equal to $3$, such that each constraint also has arity at most $3$ (i.e. each constraint function $\phi \in \Fam$ depends on at most $3$ variables). 
\begin{restatable}{lemma}{LemmaReducingDegree}\label{lemma:reducing-degree}
    If there exists a finite constraint family $\Fam$ of degree $d \geq 3$ such that we can solve $\maxcspk(\Fam)$ in time $\bigO(n^{k-\varepsilon})$ for some $\varepsilon>0$, then there exists 
    a finite constraint family $\Fam'$ of degree $3$, such that each constraint $\phi\in \Fam'$ has arity at most $3$ and such that $\maxcspk(\Fam')$ can be solved in time $\bigO(n^{k-\varepsilon'})$ for some $\varepsilon'>0$.
\end{restatable}
\begin{proof}
    Let $\phi\in \Fam$ be a function of degree $d\geq 3$ and arity $r$.
    It is easy to see that there exists a boolean function $\phi'(x_1,x_2,x_3) := \phi(x_1,x_2,x_3,x_{i_1}\dots, x_{i_{r-3}})$, where each $x_{i_j}$ is either $x_1,x_2$, or $x_3$, such that $\deg(\phi')=3$.
    Let $\Fam':=\{\phi'\}$.
    We remark that by a simple modification of the approach from \cite{KunnemannM20}, we can show that if $\maxcspk(\Fam)$ is solvable in $\bigO(n^{k-\varepsilon})$ for some $\varepsilon>0$, then also $\maxcspk(\Fam')$ is solvable in $\bigO(n^{k-\varepsilon'})$ for some $\varepsilon'>0$.
\end{proof}

Equipped with the previous lemma, it thus suffices to show the hardness of the $\maxcspk(\Fam)$ for constraint families $\Fam$ of degree equal to $3$ such that every constraint $\phi\in \Fam$ has arity at most $3$.

In particular, we prove that for any such family, given a $k$-partite regular $3$-uniform hypergraph $G$, we can construct an instance $\Phi$ of $\maxcspk(\Fam)$ and integer $\ell$ such that there exists an assignment $a$ with $\Phi(a) \geq \ell$ if and only if $G$ contains a clique.
Take an arbitrary function $\phi\in \Fam$ such that $\deg(\phi) =3$ with characteristic polynomial $f$, and define the set \[\phi_{\rm sym}(x,y,z) := \{\phi(\pi(x), \pi(y), \pi(z))\mid \pi \in S_3,\}\] and the polynomial \[\fsym(x,y,z) := \sum_{\pi\in S_3}f(\pi(x), \pi(y),\pi(z)).\]
\begin{observation}\label{obs:sym-polynomial}
    Let $\Phi$ be any instance of $\maxcspk(\Fam)$ and construct the instance $\Phi' = \Phi \cup \phisym(x,y,z)$. 
    Then the corresponding polynomial is $P_{\Phi'} = P_\Phi + \fsym$.
\end{observation}
 Moreover, the polynomial $\fsym$ is symmetric, i.e. can be written as \begin{equation}\label{eq:fsym}
 \fsym(x,y,z) = \alpha \cdot xyz + \beta (xy+xz+yz) + \gamma(x+y+z) + \delta,   
 \end{equation} 
 with $\alpha\ne 0$.
 Thus, instead of the constraint $\phi$, we use the symmetrized set $\phi_{\rm sym}$ instead and its characteristic polynomial given by $\fsym$. 


 We remark that our reduction will depend on the coefficients of this multilinear polynomial.
We first consider the coefficient $\alpha$.
Let $G = (V_1, \dots, V_k, E_G)$ be a $k$-partite $(2,\mu)$-regular $3$-uniform hypergraph, with $|V_i| = \dots = |V_k| = n$ and without loss of generality assume that $\mu = (k-2)q\cdot n$, where $q\in (0,1)$.
Let $G' = (V_1\cup \dots\cup V_k, E_{G'})$, with the edges defined as follows.
If $\alpha >0$, let $E_{G'}=E_G$.
Otherwise, if $\alpha<0$, take $E_{G'}$ as the complement of $E_G$, respecting the $k$-partition (i.e. for pairwise distinct $i,j,\ell$, and $v_i\in V_i, v_j\in V_j, v_\ell\in V_\ell$, let $\{v_i,v_j, v_\ell\}\in E_{G'}$ if and only if $\{v_i,v_j, v_\ell\}\notin E_G$).
We can observe that $G'$ is a $k$-partite $(2,\lambda)$-regular $3$-uniform hypergraph with $\lambda = \Omega(n)$.
\begin{observation}\label{obs:regularity-of-G'}
    $G'$ is a $k$-partite $(2,\lambda)$-regular $3$-uniform hypergraph, where $\lambda$ is as follows:
    \begin{enumerate}
        \item If $\alpha>0$, then $\lambda = (k-2)q\cdot n$.
        \item If $\alpha<0$, then $\lambda = (k-2)(1-q)\cdot n$.
    \end{enumerate}
\end{observation}
\medskip
We say that an assignment $a$ of variables in $X$ is \emph{$k$-partite} if for each vertex part $V_i$ in $G$, $a$ sets precisely one variable $v_i\in V_i$ to $1$ and all the others to $0$. 
We now consider the sign of the coefficient $\beta$, and make a case distinction based on this value.
\subparagraph*{{Case 1: $\beta>0$}} 
\begin{construction}\label{construction:b>0}
    Let $X = V(G')$ be the set of $kn$ variables.
    Now for each edge $\{u,v,w\}$ in $G'$, add to $\Phi$ the constraint $\phi(u,v,w)$.
    No other constraints are added to $\Phi$.
\end{construction}
Our goal is to prove that if there exists a $k$-clique in the hypergraph $G$, then any weight-$k$ assignment $a$ that maximizes $\Phi(a)$ sets the vertices in $G$ that correspond to this $k$-clique to $1$ and all other vertices to $0$. 
The strategy that we take to prove this is to first prove that any weight-$k$ assignment that maximizes $\Phi(a)$ has to be $k$-partite. We then prove that among all the $k$-partite assignments, the one that maximizes the value of $\Phi(a)$ corresponds to a $k$-clique in $G$ (i.e. either a $k$-clique, or an independent set in $G'$, depending on sign of $\alpha$). 
\begin{restatable}{lemma}{LemmaSolutionAnyAssignmentI}\label{lemma:solution-any-assignment}
    Let $a$ be any weight-$k$ assignment of the variables in $X$ and let $S$ be the set of vertices that correspond to the variables set to $1$ by $a$.
    Denote by $k_i$ the value $|S\cap V_i|$, and let $m(S)$ denote the number of edges in the induced subhypergraph $G'[S]$.
    Then we have that 
    \[
    \Phi(a) = m(S)\alpha + \sum_{(i,j)\in \binom{[k]}{2}} k_ik_j \lambda \beta + c,
    \]
    where $c$ does not depend on the assignment $a$.
\end{restatable}
\begin{proof}
    Using Observation \ref{lemma:char-polynomial} and Observation \ref{obs:sym-polynomial}, we have the following equality
    \[
        \Phi(a) = \sum_{\{u_1,u_2,u_3\}\in E(G')} \fsym(a(u_1),a(u_2),a(u_3)).
    \]
    Recall that in Equation \ref{eq:fsym} we have $\fsym(x,y,z) = \alpha \cdot xyz + \beta (xy+xz+yz) + \gamma(x+y+z) + \delta$. Plugging this in the equality above, we get
    \begin{align*}
        \Phi(a) &= \sum_{\{u_1,u_2,u_3\}\in E(G')} \alpha \cdot a(u_1)a(u_2)a(u_3) 
                    \\ &+ \sum_{\{u_1,u_2,u_3\}\in E(G')}\beta \left(a(u_1)a(u_2)+a(u_1)a(u_3)+a(u_2)a(u_3)\right) 
                    \\ & + \sum_{\{u_1,u_2,u_3\}\in E(G')}  \gamma\left(a(u_1)+a(u_2)+a(u_3)\right) 
                    \\ & + \sum_{\{u_1,u_2,u_3\}\in E(G')} \delta
    \end{align*}
    We proceed to argue the equality of each of the coefficients. 
    Clearly $\sum_{\{u_1,u_2,u_3\}\in E(G')} \delta = \delta |E(G')|$. 
    Moreover, the term $\alpha \cdot a(u_1)a(u_2)a(u_3)$ is non-zero if and only if all three of the variables $u_1,u_2,u_3$ are set to one. Hence, we can write 
    \[
    \sum_{\{u_1,u_2,u_3\}\in E(G')} \alpha \cdot a(u_1)a(u_2)a(u_3) = \sum_{\{u_1,u_2,u_3\}\in E(S)} \alpha = \alpha |E(G'[S])| = \alpha m(S)
    \]
    as desired. 
    Next, by rearranging the summands, and using regularity of $G'$ (namely, for any pair $u,v\in V(G')$, the equality $\deg(u)= \deg(v)$ holds), we obtain:
    \begin{align*}
        \sum_{\{u_1,u_2,u_3\}\in E(G)}  \gamma\left(a(u_1)+a(u_2)+a(u_3)\right)  &= \gamma\sum_{\{u_1,u_2,u_3\}\in E(G)}  \left(a(u_1)+a(u_2)+a(u_3)\right)  \\
        &= \gamma|S|\deg(u) \tag{\textasteriskcentered}\\
        &= \gamma k\deg(u),
    \end{align*}
    where the equality $(*)$ follows by noticing that for each vertex $u\in S$, the term $a(u)$ appears precisely $\deg(u)$ many times in the expression above with the value $1$ (and for any $v\not\in S, a(v)=0$).
    Finally, recall that $G'$ is $(2,\lambda)$-regular $k$-partite hypergraph, i.e.
    for any pair of variables that correspond to vertices $v_i\in V_i, v_j\in V_j$ respectively, for $i\neq j$, there are precisely $\lambda$ many edges $e\in E(G')$ that contain both $v_i$ and $v_j$. 
    Thus, a similar computation as above shows that 
    \[
    \sum_{\{u_1,u_2,u_3\}\in E(G')}\beta \left(a(u_1)a(u_2)+a(u_1)a(u_3)+a(u_2)a(u_3)\right) = \sum_{(i,j)\in \binom{[k]}{2}}k_ik_j\lambda \beta. \qedhere
    \]
\end{proof}

As an immediate consequence of the last lemma, we get the number of satisfied constraints by any $k$-partite assignment $a$, just by plugging in $k_i=1$ for each $i\in [k]$.
\begin{corollary}\label{cor:solution-k-partite-assignment}
    Let $a$ be any $k$-partite assignment of the variables in $X$ and let $S, m(S)$ be defined as above and $u$ be an arbitrary vertex in $G'$.
    Then
    \[
    \Phi(a) = m(S)\alpha + \binom{k}{2} \lambda \beta + k \deg(u) \gamma + |E(G)| \delta,
    \]
\end{corollary}
We now prove that any $k$-partite assignment is favorable over any non-partite assignment.
\begin{restatable}{lemma}{LemmakPartiteOptimalityI}\label{lemma:optimality-of-k-partite-assignment}
    Let $a$ be a $k$-partite assignment and let $a'$ be any non-$k$-partite assignment of weight $k$.
    Let $S$ and $S'$ be the sets of vertices that correspond to the variables set to $1$ by $a$ and $a'$ respectively.
    Then $\Phi(a)>\Phi(a')$.
\end{restatable}
\begin{proof}
    Consider the value $\Phi(a)-\Phi(a')$. By Lemma \ref{lemma:solution-any-assignment} and Corollary \ref{cor:solution-k-partite-assignment} we compute:
    \begin{align*}
        \Phi(a)-\Phi(a') &= \left(m(S)\alpha + \binom{k}{2} \lambda \beta \right) - \left( m(S')\alpha + \sum_{(i,j)\in \binom{[k]}{2}} k_ik_j \lambda \beta \right) \\
        & = \left(m(S)-m(S')\right)\alpha + \left(\binom{k}{2}-\sum_{(i,j)\in \binom{[k]}{2}} k_ik_j\right)\lambda\beta.
    \end{align*}
    Let us consider the value of $\binom{k}{2}-\sum_{(i,j)\in \binom{[k]}{2}} k_ik_j$. In particular, we show that if $a'$ is not a $k$-partite assignment, this value is strictly positive.
    \begin{claim}
        For each $i\in [k]$, let $k_i$ be such that $\sum_i k_i = k$. Then $\sum_{(i,j)\in \binom{[k]}{2}}k_ik_j\leq \binom{k}{2}$, with equality if and only if $k_1=\dots = k_k = 1$.
    \end{claim}
    \begin{proof}[Proof (of the claim).]
        Recall that
        \[
        \left(\sum_{i=1}^kk_i\right)^2 = \sum_{i=1}^kk_i^2 + 2\sum_{i,j\in \binom{[k]}{2}}k_ik_j
        \]
        Plugging in $\sum_{i=1}^kk_i=k$ and rearranging, we get
        \[
        \sum_{i,j\in \binom{[k]}{2}}k_ik_j = \left(k^2 - \sum_{i=1}^kk_i^2\right)\Big/{2}.
        \]
        By Quadratic Mean -- Arithmetic Mean Inequality it holds that:
        \[
            \sqrt{\frac{\sum_{i=1}^k k_i^2}{k}}\geq \frac{\sum_{i=1}^kk_i}{k},
        \]
        with equality if and only if $k_1=\dots = k_k$.
        We can notice that $\frac{\sum_{i=1}^kk_i}{k} = 1$ and hence this inequality implies that $\sum_{i=1}^k k_i^2\geq k$.
        Plugging this back into our expression above, we have the following chain of inequalities:
        \[
            \sum_{i,j\in \binom{[k]}{2}}k_ik_j = \left(k^2 - \sum_{i=1}^kk_i^2\right)\Big/{2} \leq \left(k^2 - k\right)\big/{2} = \binom{k}{2},
        \]
        with equality if and only if $k_1=\dots = k_k=1$.
        \renewcommand{\qedsymbol}{$\vartriangleleft$}
    \end{proof}
    Using the previous claim, we can conclude that there exists an $\varepsilon>0$ (independent on $n$), such that
    \begin{align*}
        \Phi(a)-\Phi(a') &\ge \left(m(S)-m(S')\right)\alpha + \varepsilon\lambda\beta.
    \end{align*}
    Recall that $\alpha$ is the leading coefficient of the polynomial $\fsym$ which is obtained by summing six multilinear polynomials that correspond to a constraint function $\phi$. As such, clearly $\alpha\in \bigO(1)$ holds. Moreover, we also observe that $m(S)-m(S')\geq -\binom{k}{3}$. On the other hand by
    Observation \ref{obs:regularity-of-G'}
    the value of $\lambda$ in the expression above satisfies $\lambda = \Omega(n)$, and in particular, since $\varepsilon, \beta$ are positive constants independent on $n$, we have $\varepsilon\lambda\beta = \Omega(n)$.
    Hence, plugging all this in, we can write
    \begin{align*}
        \Phi(a)-\Phi(a') &\ge \Omega(n) - \binom{k}{3}\bigO(1) = \Omega(n)>0. \qedhere
    \end{align*}
\end{proof}

We are now ready to prove that if we can efficiently solve the instance $\maxcspk(\Fam)$, we can also efficiently decide if the given hypergraph $G$ contains a $k$-clique.
\begin{restatable}{lemma}{LemmaCaseIMainLemma}\label{lemma:case-1-main-lemma}
    There exists a weight-$k$ assignment $a$ and an integer $\tau$, such that $\Phi(a)\geq \tau$ if and only if $G$ contains a clique of size $k$. Moreover, $\tau$ can be computed in constant time.
\end{restatable}
\begin{proof} 
    Let $a$ be an assignment that maximizes the value $\Phi(a)$.
    By Lemma \ref{lemma:optimality-of-k-partite-assignment}, we can assume that $a$ is $k$-partite.
    Recall that by $S$ we denote the set of vertices that correspond to the variables that $a$ sets to $1$.
    Then by Corollary \ref{cor:solution-k-partite-assignment}, we have 
    \[
     \Phi(a) = m(S)\alpha + \binom{k}{2} \lambda \beta + k \deg(u) \gamma + |E(G')| \delta.
    \]
    Note that the value $c:=\binom{k}{2} \lambda \beta + 3 k \deg(u) \gamma + |E(G)| \delta$ is equal for all $k$-partite assignments, i.e. does not depend on the choice of $a$.
    We now make a distinction between two cases based on the sign of $\alpha$.
    \begin{enumerate}
        \item If $\alpha>0$, set $\tau := \binom{k}{3}\alpha + c$. Now clearly, any assignment $a$ satisfies $\Phi(a)\ge \tau$ if and only if $m(S) \geq \binom{k}{3}$, that is if and only if the vertices in $S$ form a clique in $G'$.
        Recall that for $\alpha >0$ we have $G=G'$, thus proving that there exists an assignment $a$ with $\Phi(a)\ge \tau$ if and only if $G$ contains a clique of size $k$. 
        \item If $\alpha<0$, set $\tau := c$. Now clearly, the inequality $\Phi(a)\ge \tau$ holds if and only if $m(S) \leq 0$, that is if and only if the vertices in $S$ form an independent set in $G'$.
        Note that by construction, when $\alpha<0$, the vertices $v_1\in V_1,\dots, v_k\in V_k$ form an independent set in $G'$ if and only if they form a clique in $G$.
        In particular, this shows that there exists an assignment $a$ satisfying $\Phi(a)\ge \tau$ if and only if $G$ contains a $k$-clique. \qedhere
    \end{enumerate}
\end{proof}
\subparagraph*{{Case 2: $\beta<0$}}
\begin{construction}\label{construction:b<0}
    Let $X = V(G')$ be the set of $kn$ variables.
    For each edge $\{u,v,w\}$ in $G'$, add $\phi(u,v,w)$ to $\Phi$.
    Furthermore, for each $i\in [k]$ and each pair of distinct vertices $u,v\in V_i$, add $\phi(u,v,w)$ to $\Phi$ for any $w\in V(G')\setminus\{u,v\}$.
\end{construction}
Following the general structure of the first case, we first count the number of constraints satisfied by some assignment $a$ and then show that any optimal assignment is $k$-partite.
\begin{restatable}{lemma}{LemmaCaseIIPolynomial}\label{lemma:case2-any-assignment}
    Let $a$ be any weight-$k$ assignment and let $S$ be the set of vertices that correspond to the variables set to $1$ by $a$.
    Denote by $k_i$ the value $|S\cap V_i|$, and let $m(S)$ denote the number of edges in the induced subhypergraph $G'[S]$.
    Then the following equality holds:
    \begin{align*}
    \Phi(a) &= \alpha\cdot \left(m(S) + \sum_{i\in [k]}\binom{k_i}{3} + \sum_{i,j\in\binom{[k]}{2}}\left(\binom{k_i}{2}k_j + \binom{k_j}{2}k_i\right)\right) 
    \\ &+ \beta \cdot\left(\sum_{(i,j)\in \binom{[k]}{2}} k_ik_j (\lambda + 2n-k_i-k_j) + \sum_{i\in [k]} \left(\binom{k_i}{2}(kn-2)\right)\right)  + c
    \end{align*}
    where $c$ is a value that does not depend on the choice of $a$.
\end{restatable}
\begin{proof}
    Using Observation \ref{lemma:char-polynomial} and Observation \ref{obs:sym-polynomial}, we have the following equality
        \begin{align*}
            \Phi(a) &= \left(\sum_{\{u_1,u_2,u_3\}\in E(G')} \fsym(a(u_1),a(u_2),a(u_3))\right) \\
            &+ \left( \sum_{i\in[k]}\sum_{\substack{\{u_1,u_2\}\in \binom{V_i}{2}\\u_3\in V(G')\setminus\{u_1,u_2\} }} \fsym(a(u_1),a(u_2),a(u_3))\right),
        \end{align*}
        Notice that we already computed the value of $\sum_{\{u_1,u_2,u_3\}\in E(G')} \fsym(a(u_1),a(u_2),a(u_3))$ in Lemma \ref{lemma:solution-any-assignment}, hence it only remains to argue the coefficients in the second part of the equality.
        For simplicity we will write $\sum_{u_1,u_2,u_3}$ instead of $\sum_{\substack{\{u_1,u_2\}\in \binom{V_i}{2}\\u_3\in V(G')\setminus\{u_1,u_2\} }}$.
        Similarly as before, we can express:
        \begin{align*}
            \sum_{i\in[k]}\sum_{u_1, u_2, u_3} \fsym(a(u_1),a(u_2),&a(u_3)) = \sum_{i\in[k]}\sum_{u_1, u_2, u_3} \alpha \cdot a(u_1)a(u_2)a(u_3) 
                        \\ &+ \sum_{i\in[k]}\sum_{u_1, u_2, u_3}\beta \left(a(u_1)a(u_2)+a(u_1)a(u_3)+a(u_2)a(u_3)\right) 
                        \\ & + \sum_{i\in[k]}\sum_{u_1, u_2, u_3}  \gamma\left(a(u_1)+a(u_2)+a(u_3)\right) 
                        \\ & + \sum_{i\in[k]}\sum_{u_1, u_2, u_3} \delta
        \end{align*}
        
        It is easy to see that the coefficient of $\delta$ does not depend on the choice of the assignment. Moreover, since each vertex $u$ is contained in the same number of clauses, clearly also the coefficient of $\gamma$ is equal for any assignment $a$ of weight $k$.
        
        We compute the remaining two coefficients.
        Recall that for any choice of $u_1,u_2,u_3$, the value $a(u_1)a(u_2)a(u_3) = 1$ if and only if $\{u_1,u_2,u_3\}\subseteq S$. 
        We want to compute the number of such choices where at least two variables stem from the same set.
        There are two possibilities. Either $u_1,u_2,u_3\in V_i$ (for some $i\in[k]$), in which case we have precisely $\binom{k_i}{3}$ many valid choices of vertices $u_1,u_2,u_3$ that satisfy $a(u_1)a(u_2)a(u_3) = 1$.
        Alternatively, we have precisely two vertices $u_1,u_2$ are contained in $V_i$ and one vertex in $V_j$ for $j\neq i$, in which case we have $\binom{k_i}{2}$ many choices for $u_1,u_2$ and $k_j$ many choices for $u_3$ such that $a(u_1)a(u_2)a(u_3) = 1$.
        Hence we obtain the desired equality:
        \[
        \sum_{i\in[k]}\sum_{u_1, u_2, u_3} \alpha \cdot a(u_1)a(u_2)a(u_3) = \alpha\cdot \left(\sum_{i\in [k]}\binom{k_i}{3} + \sum_{i,j\in\binom{[k]}{2}}\left(\binom{k_i}{2}k_j + \binom{k_j}{2}k_i\right)\right) 
        \]
        Finally, to calculate the coefficient of $\beta$, let $u_1,u_2\in S\cap V_i$ be arbitrary. 
        Then for any choice of $u_3\in V(G')\setminus\{u_1,u_2\}$, it holds that $\phisym(u_1,u_2,u_3)\in \Phi$. This contributes to the coefficient of $\beta$ the value of $\sum_{i\in [k]} \left(\binom{k_i}{2}k(|V_i|-2)\right)$.
        But notice that so far we only considered the case $u_1,u_2\in V_i\cap S$, while by looking at the domain of our sum, it is clear that we want to compute the contribution for all pairs $u_1,u_2\in V_i$.
        Hence, we additionally want to compute the number of triples $u_1,u_2,u_3$ with $u_1\in V_i\cap S, u_2\in V_i\setminus S, u_3\in V_j\cap S$ for $j\neq i$.
        It is easy to see that there are $\sum_{(i,j)\in \binom{[k]}{2}}k_ik_j(|V_i| + |V_j|-k_i - k_j) = \sum_{(i,j)\in \binom{[k]}{2}}k_ik_j(2n -k_i - k_j)$ such triples.
        
        In particular, we get the equality
        \begin{align*}
            &\sum_{i\in[k]}\sum_{\substack{\{u_1,u_2\}\in \binom{V_i}{2}\\u_3\in V(G')\setminus\{u_1,u_2\} }}a(u_1)a(u_2)+a(u_1)a(u_3)+a(u_2)a(u_3)\\
            = &\left(\sum_{i\in [k]} \binom{k_i}{2}(kn-2)\right) + \left(\sum_{i,j\in\binom{k}{2}}k_ik_j\left(2n-k_i-k_j\right)\right). \qedhere
        \end{align*}

    \end{proof}

We can now argue that for any $k$-partite assignment $a$ and any assignment $a'$ that is not $k$-partite, the inequality $\Phi(a)>\Phi(a')$ is satisfied.
\begin{restatable}{lemma}{LemmakPartiteOptimalityII}\label{lemma:case2-optimality-of-k-partite-assignment}
    Let $a$ be any $k$-partite assignment of the constraint in $\Phi$ and let $a'$ be any non-$k$-partite assignment of weight $k$.
    Let $S$ and $S'$ be the sets of vertices that correspond to the variables set to $1$ by $a$ and $a'$ respectively.
    Then $\Phi(a)>\Phi(a')$.
\end{restatable}
\begin{proof}
    Denote by $k_i$ the value $|S'\cap V_i|$.
    Similarly as in proof of Lemma \ref{lemma:optimality-of-k-partite-assignment}, we consider the value $\Phi(a)-\Phi(a')$ and write:

    \begin{align*}
        &\Phi(a)-\Phi(a')  = \alpha\cdot \left(m(S) - m(S') - \sum_{i\in [k]}\binom{k_i}{3} - \sum_{i,j\in\binom{[k]}{2}}\left(\binom{k_i}{2}k_j + \binom{k_j}{2}k_i\right)\right) 
        \\ &+ \beta \cdot\left(\binom{k}{2} (\lambda + 2n-2)-\sum_{(i,j)\in \binom{[k]}{2}} k_ik_j (\lambda + 2n-k_i-k_j) - \sum_{i\in [k]} \left(\binom{k_i}{2}(kn-2)\right)\right) \tageq\label{full-eq}
    \end{align*}
    Let $q$ be as in Observation \ref{obs:regularity-of-G'} and define $p:=\max\{q,1-q\}$. We rewrite the term $kn-2$ as:
    \begin{align*}
        kn-2 &= (k-2)n + 2n - 2 - (k-2)(1-p)n + (k-2)(1-p)n\\
        &=(k-2)pn + 2n-2 + (k-2)(1-p)n \\
        &\geq \lambda + 2n-2 + (k-2)(1-p)n & \text{(by Observation \ref{obs:regularity-of-G'})}
    \end{align*}
    Hence, we can write:
    \begin{align*}
    \sum_{i\in [k]} \left(\binom{k_i}{2}(kn-2)\right) &\geq \sum_{i\in [k]} \left(\binom{k_i}{2}(\lambda + 2n-2 + (k-2)(1-p)n)\right) 
    \\ & = \sum_{i\in [k]} \left(\binom{k_i}{2}(\lambda + 2n-2)\right) + \sum_{i\in[k]}\binom{k_i}{2}(k-2)(1-p)n\\
    & \geq \sum_{i\in [k]} \left(\binom{k_i}{2}(\lambda + 2n-2)\right) + (k-2)(1-p)n\\
    &\geq \sum_{i\in [k]} \left(\binom{k_i}{2}(\lambda + 2n-2)\right) + (1-p)n + k\binom{k}{2}, \tageq\label{bound-first-term}
    \end{align*}
    where in the second to last inequality we use that $a'$ is not $k$-partite (i.e. for at least one $i\in[k]$, $k_i\geq 2$) and in the last inequality we use that $k-2\geq 2$ and $(1-p)n = \Omega(n) > k \binom{k}{2}$ (the term $k\binom{k}{2}$ is chosen for convenience as will be evident in a moment).
    For the other term, we can write 
    \begin{align*}
    \sum_{(i,j)\in \binom{[k]}{2}} k_ik_j (\lambda + 2n-k_i-k_j)  &\geq \sum_{(i,j)\in \binom{[k]}{2}} k_ik_j (\lambda + 2n-k)
    \\ &= \left(\sum_{(i,j)\in \binom{[k]}{2}} k_ik_j (\lambda + 2n) \right)- k\binom{k}{2}\\
    &> \left(\sum_{(i,j)\in \binom{[k]}{2}} k_ik_j (\lambda + 2n-2)\right) - k\binom{k}{2} \tageq \label{bound-second-term}
    \end{align*}
    Before plugging this bound in the expression above, we first prove one more equality that will be particularly useful in bounding the coefficient of $\beta$.
    \begin{claim}\label{claim:sum-of-kis}
        Let $k_i\geq 0$ be integers satisfying $\sum_{i\in[k]}k_i = k$. Then \[
        \sum_{i,j\in \binom{[k]}{2}}k_ik_j + \sum_{i\in [k]}\binom{k_i}{2} = \binom{k}{2}
        \]
    \end{claim}
    \begin{proof}[Proof (of the claim).]
        \renewcommand{\qedsymbol}{$\vartriangleleft$}
        Using $k^2 = (\sum_{i\in[k]} k_i)^2 = \sum_{i\in[k]}k_i^2 + 2\sum_{i,j\in\binom{[k]}{2}}k_ik_j$, we can write $\sum_{i,j\in \binom{[k]}{2}}k_ik_j = k^2 - \sum_{i\in [k]}k_i^2$. Plugging this in, we get:
        \begin{align*} 
            \sum_{i,j\in \binom{[k]}{2}}k_ik_j + \sum_{i\in [k]}\binom{k_i}{2} &= \left(k^2 - \sum_{i\in [k]}k_i^2\right)\Big/2 + \sum_{i\in[k]}\frac{k_i(k_i-1)}{2} \\ 
            &= \frac{k^2}{2} + \sum_{i\in[k]}\frac{k_i(k_i-1)-k_i^2}{2} \\
            &= \frac{k^2}{2} - \sum_{i\in[k]}\frac{k_i}{2} \\
            &= \frac{k^2}{2} - \frac{k}{2}\\
            &= \binom{k}{2}.
        \end{align*}
    \end{proof}
    Now we can plug in the Inequality (\ref{bound-first-term}) and (\ref{bound-second-term}) into the full expression (i.e. Equation (\ref{full-eq})). Let us focus on the negative terms in the expression for the coefficient of $\beta$:
    \begin{align*}
    \mathmakebox[\mathindent][l]{
        \left(\sum_{(i,j)\in \binom{[k]}{2}} k_ik_j (\lambda + 2n-k_i-k_j)\right) + \left(\sum_{i\in [k]} \binom{k_i}{2}(kn-2)\right) }
        \\ &\geq \left(\sum_{(i,j)\in \binom{[k]}{2}} k_ik_j (\lambda + 2n-2)\right) - k\binom{k}{2} + \left(\sum_{i\in [k]} \binom{k_i}{2}(\lambda + 2n-2)\right) + (1-p)n + k\binom{k}{2} \\
        & = \left(\sum_{(i,j)\in \binom{[k]}{2}} k_ik_j (\lambda + 2n-2)\right) + \left(\sum_{i\in [k]} \binom{k_i}{2}(\lambda + 2n-2)\right) + (1-p)n  \\
        & = \left(\lambda + 2n-2\right)\left(\sum_{(i,j)\in \binom{[k]}{2}} k_ik_j + \sum_{i\in [k]} \binom{k_i}{2}\right)  + (1-p)n\\
        & = \binom{k}{2}\left(\lambda + 2n-2\right) + (1-p)n, \tageq\label{eq:final-bound}
    \end{align*}
    where the last equality follows by Claim \ref{claim:sum-of-kis}.
    Let $\Phi(a)[\beta]$ denote the coefficient of $\beta$ in the expression of $\Phi(a)$.
    If we plug in our bound from Inequality (\ref{eq:final-bound}) to Equation (\ref{full-eq}), we obtain
    \begin{align*}
        \Phi(a)[\beta] \leq  \binom{k}{2} (\lambda + 2n-2) - \left(\binom{k}{2}\left(\lambda + 2n-2\right) + (1-p)n\right) = -(1-p)n. 
    \end{align*}
    Moreover, by noticing that $\Phi(a)[\alpha] \leq f(k)\ll(1-p)n$, (for some function $f$), and recalling that $\beta<0$, we get
     \begin{align*}
          \Phi(a)-\Phi(a') & \geq \alpha \cdot f(k) - \beta \cdot (1-p)n >0 \qedhere
     \end{align*}
\end{proof}
We now focus on $k$-partite assignments only, and show that among all $k$-partite assignments, the one that maximizes the value of $\Phi(a)$ corresponds to a $k$-clique in $G$ (if it exists).

\begin{lemma}\label{lemma:case-2-main-lemma}
    There exists a weight-$k$ assignment $a$ and an integer $\tau$ such that $\Phi(a)\geq \tau$ if and only if $G$ contains a clique of size $k$. Moreover, $\tau$ can be computed in constant time.
\end{lemma}
\begin{proof}
    Let $a$ be an assignment of weight $k$ that maximizes the value $\Phi(a)$. By the previous lemma, $a$ is $k$-partite, thus $k_i=1$ for $i \in [k]$ and with Lemma \ref{lemma:case2-any-assignment}, satisfies:
    \[
    \Phi(a) = \alpha \cdot m(S)  
    + \beta \cdot\binom{k}{2} (\lambda + 2n-2)
    + c.
    \]
    We now observe that the part $d:= \beta \cdot\binom{k}{2} (\lambda + 2n-2)
    + c$ is independent of the choice of $k$-partite assignment $a$.
    In particular, $\Phi(a) = \alpha \cdot m(S) + d$. 
    The result now follows by precisely the same argument as in the proof of Lemma \ref{lemma:case-1-main-lemma}.
\end{proof}
\subparagraph*{{Case 3: $\beta=0$}} 
For Case 3, we employ the constructions used in the previous cases
\begin{construction}\label{construction:b=0}
    Let $X = V(G')$ be the set of $kn$ variables.
    We make a case distinction based on the value of $\alpha$:
    \begin{enumerate}
        \item If $\alpha>0$, then let $\Phi$ be as in Construction \ref{construction:b>0}.
        \item If $\alpha<0$, then let $\Phi$ be as in Construction \ref{construction:b<0}.
    \end{enumerate}
\end{construction}
\medskip
For any weight $k$-assignment, we can compute the value $\Phi(a)$, by plugging in $\beta=0$ into the corresponding construction.
\begin{lemma}
    Let $\Phi$ be a $\maxcspk(\Fam)$ instance as above, and let $a$ be any assignment of weight $k$. 
    Then the following equalities hold.
    \begin{enumerate}
        \item If $\alpha>0$, 
        \[
            \Phi(a) = m(S)\alpha + c  .   
        \]
        \item  If $\alpha<0$,
        \[
        \Phi(a) = \alpha\cdot \left(m(S) + \sum_{i\in [k]}\binom{k_i}{3} + \sum_{i,j\in\binom{[k]}{2}}\left(\binom{k_i}{2}k_j + \binom{k_j}{2}k_i\right)\right) +  c,
        \]
        where $c$ is a value that does not depend on the choice of the weight-$k$ assignment $a$.
    \end{enumerate}
\end{lemma}
We can now show that if we can efficiently find the assignment $a$ that maximizes the value $\Phi(a)$, then we can also efficiently decide if the hypergraph $G$ contains a clique of size $k$.
\begin{restatable}{lemma}{LemmaCaseIIIMain}\label{lemma:main-lemma-case-3}
    There exists a number $\tau$ and an assignment $a$ of weight $k$, such that $\Phi(a)\geq \tau$ if and only if the hypergraph $G$ contains a clique of size $k$. Moreover, we can compute $\tau$ in constant time.
\end{restatable}
\begin{proof}
    Let us first consider the case $\alpha>0$. Then we set $\tau:=\binom{k}{3}\alpha + c$.
    By the previous lemma, for any assignment $a$, the inequality $\Phi(a)\geq \tau$ holds if and only if the vertices in $S$ form a $k$-clique in the hypergraph $G'$, and since $\alpha>0$, we have $G'=G$, proving the desired.
    On the other hand, if $\alpha<0$, we set $\tau:=c$. We first notice that this value is attainable by $\Phi(a)$ only if $a$ is a $k$-partite assignment.
    Now for any $k$-partite assignment $a$, we have
    $\Phi(a) = \alpha\cdot m(S) + c$,
    and thus $\Phi(a)\geq \tau$ if and only if $m(S) = 0$, i.e. if the vertices in $S$ form an independent set in $G'$. 
    Furthermore, we recall that since $\alpha<0$, each $k$-partite independent set in $G'$ corresponds bijectively to a $k$-clique in $G$.
    
    This proves that regardless of the value of $\alpha$, we can find in constant time the value $\tau$ such that there exists an assignment $a$ satisfying $\Phi(a)\geq \tau$, if and only if $G$ contains a clique of size $k$.
\end{proof}

We can finally prove Theorem \ref{th:CSP-hardness}.
\begin{proof}[Proof of Theorem \ref{th:CSP-hardness}]
    Assume that there exists a family $\Fam$ with $\deg(\Fam)\geq 3$ and that there exists an algorithm solving $\maxcspk(\Fam)$ in time $\bigO(n^{k-\varepsilon})$. 
    With our previous observation, detailed in the full paper, we construct a family $\Fam'$ of degree equal to $3$ such that $\maxcspk(\Fam')$ is solvable in time $\bigO(n^{k-\varepsilon'})$.
    Then given a regular $3$-uniform hypergraph $G$, we apply the corresponding construction (depending on the value of $\beta$) to reduce to construct an instance $\Phi$ of $\maxcspk(\Fam')$ that consists of $kn$ variables and $\bigO(m)$ constraints.
    Now, by Lemmas \ref{lemma:case-1-main-lemma}, \ref{lemma:case-2-main-lemma}, \ref{lemma:main-lemma-case-3}, we can decide if $G$ contains a clique of size $k$, by computing an assignment $a$ that maximizes the value $\Phi(a)$ in time $\bigO(n^{k-\varepsilon'})$, thus refuting the $3$-uniform $k$-Hyperclique Hypothesis.
\end{proof}

\subsection{Algorithmic Results}
Algorithms for $\maxcspk(\Fam)$ have already been studied from a parametrized viewpoint. In \cite{Cai08}, Cai showed for the three families $\{\xor\},\{\OR\}$ and $\{\AND\}$, each of degree $2$,
that they can be solved in time $\bigO(f(k) n^{\omega \lfloor k/3 \rfloor + 1 + k \mod 3})$ by reducing them to finding a maximum-weight triangle. 
In this section, we extend this with an algorithm for families of all degrees, generalizing Cai's running time for degree $2$ families.

First, we take a look at the unconstrained version $\maxcsp(\Fam)$, that asks for assignments of any weight.
Ryan Williams first presented an algorithm for this version in \cite{Wil07}, where $\Fam$ was of degree $2$, reducing the problem to $k$-clique.
Subsequently, this result was generalized to $h$-uniform $\ell$-hyperclique for all families of degree $h>2$ by Lincoln, {Vassilevska Williams} and Williams in \cite{LincolnWW18}.
The approach and techniques employed can be transferred nicely into the weight $k$ setting with minor modifications.
We give a short summary of the general algorithm, explain the modifications to accommodate the constrained setting and what changes in the analysis, 
to obtain algorithms for all constraint families of degree $h$.

\begin{proof}[Proof of Theorem \ref{thm:algorithms}]
There are only few modifications needed to adapt the general algorithm described in \cite{LincolnWW18} to our setting of considering only weight $k$ assignments.
Let $\Phi$ be an instance of $\maxcspk(\Fam)$ of degree $h$ over variables $X = \{x_1,\dots,x_n\}$.
For a constraint function $\phi_q \in \Fam$ let $A_q = \{C_i \in \Phi \mid C_i = \phi_q(x_{i_1},\dots,x_{i_r})\}$ be the set of clauses that it is used in.
Then we define the polynomial over all those applications of $\phi_q$ where $f_q$ is the corresponding characteristic polynomial as
\[
    P_q(x_y,\dots,x_n) = \sum_{C_i \in A_q} f_q(x_{i_1},\dots,x_{i_r}) 
\]
then constructs a weighted $h$-uniform $\ell$-partite hypergraph 
containing a clique of maximum weight exactly $\sum_{q=1}^{|\Fam|} P_q(x_1,\dots,x_n)$ with $h<\ell\leq k$, then reduces this to the unweighted case.

For our setting, we no longer split the variables and consider all partial assignments over these groups, but instead consider all weight $k/\ell$ assignments over $n$ variables. We argue the exact choice of $\ell$ to obtain an integer $k/\ell$ later on.
We identify each assignment with a vertex $v \in V_i$ for every $i \in [\ell]$. This creates no more than $\ell \cdot {n \choose {k/\ell}}$ vertices. 

Denote by $S(v_i)$ the variables set to $1$ under the assignment induced by $v_i$ for some $i\in [\ell]$.
Then add an edge $v_{i_1},\dots,v_{i_{h}}$ if and only if all $S(v_{i_j})$ are disjoint and all variables for $S(v_{i_j})$ are lexicographically smaller than the variables in $S(v_{i_{j+1}})$. Then this edge corresponds to the joint assignment of $v_{i_1},\dots,v_{i_{h}}$.

Consider by $M(v_{i},\dots,v_{i_{h}})$ the set of monomials $c$ over all polynomials $P_q$ such that 
all variables of the monomial $c$ are contained in $S(v_{i_1},\dots,v_{i_h}) = \bigcup_{j=1}^h S(v_{i_j})$ and $S(v_{i_1},\dots,v_{i_h})$ is the lexicographically smallest such set.
As the monomials are of degree at most $h$, every monomial is contained in some set $M$.

We then add the following weight to edge $e = \{v_{i_1},\dots,v_{i_{h}}\}$ :
\[
   W(e) = \sum_{c \in S(v_{i_1},\dots,v_{i_{h}})} c(v_{i_1},\dots,v_{i_{h}})
\]
where $c(v_{i_1},\dots,v_{i_{h}})$ is the monomial evaluated under the assignment corresponding to $e$.
The argument then follows to be the same as in the original construction: A clique $C_l$ corresponds to exactly one assignment of weight $k$ and has the weight of all monomials in $\sum_{q=1}^{|\Fam|}p_q(x_1,\dots,x_n)$, as a monomial is only accounted for in the weight of the lexicographically first edge that fully covers all its variables. Conversely, all crucial monomials will be covered by some choice of $v_{i_1},\dots,v_{i_h}$. A monomial containing a variables set to zero does not contribute to the weight in any case.

The algorithms running time crucially depends on the weights assigned to each edge.
The analysis changes only slightly. Andrew et al. showed that the coefficients of a characteristic polynomial of a boolean function of degree $h$ are bounded by $[-2^h,2^h]$.
As any assignment we consider has weight at most $k$, the weight of an edge is constituted of at most $k \choose h$ distinct monomials. 
Thus, over $|\Fam|$ many polynomials, the weight function is bounded by $[ -|\Fam|\cdot k^h 2^h,|\Fam|\cdot k^h 2^h ]$.
Our constructed graph has $\mathcal{O}(\ell \cdot n^{k/\ell})$ many vertices, 
we construct at most $\mathcal{O}\left({\ell \choose h} \cdot n^{ (k/\ell) \ell}\right)$ many edges, 
each weight can be produced in time $|\Fam| \cdot {k \choose h}$ resulting in a total construction time of $\mathcal{O}\left(|\Fam| \cdot k^h \ell^h n^k\right)$.

Finding a maximum-weight triangle can be achieved in the time needed to compute the $(\min,+)$-product over two $n\times n$ matrices (see e.g.~\cite{VassilevskaW06}). An algorithm by Zwick \cite{Zwick02} achieves this in time $\mathcal{O}(M \log M n^\omega)$ for integer weights in range $[-M,M]$.
Thus, for $h = 2$ we can choose $\ell = 3$ and obtain an algorithm for $\maxcspk(\Fam)$ that runs in $g(k) n^{\omega k/3}$ with $g(k) = |\Fam|\cdot k^h 2^h$ for $k$ divisible by $3$ and $g(k) n^{\omega (\lfloor k/3\rfloor,\lceil k/3\rceil,\lceil (k-1)/3\rceil)}$ otherwise, with slightly unbalanced partitions. Here we can employ the fastest rectangular matrix multiplication algorithm where $\omega(a,b,c)$ is the time needed to compute the product of a $n^a\times n^b$ and a $n^b \times n^c$ matrix.

For $h>2$, we can solve the constructed maximum-weight $l$-hyperclique instance with the trivial brute-force algorithm in 
time $\bigO\left(g(k) \cdot (l\cdot n^{k/l})^l\right)$. Note here, that we can always choose $\ell$ in such a way, that $k$ is divisible by~$\ell$. 

But, with the approach taken by Lincoln et al. reducing the maximum-weight $l$-hyperclique problem to the unweighted version,
we can show that any improvement in hyperclique detection translates into an improved algorithm for $\maxcspk(\Fam)$ for families of degree $h > 2$.
Their approach consists of \emph{guessing} the weights of the edges contained in the hyperclique, removing edges of a different weight. I.e., for an edge $\{u_{i_1},\dots,u_{i_h}\}$ in the $\ell$-hyperclique, guess a weight in $[-M,M]$ and delete all edges between parts $V_{i_1},\dots,V_{i_h}$ that are not of that weight. Doing this for all hyperclique edges, any $\ell$-hyperclique that is detected will correspond to an $\ell$-hyperclique of exactly the weight we guessed.

This results in a total of $\mathcal{O}\left((|\Fam| \cdot k^h 2^h)^{\ell \choose h}\right)$ instances of $h$-uniform $\ell$-hyperclique we need to solve.
Let $T(n)$ be the time needed to solve an instance of $h$-uniform $\ell$-hyperclique on a graph with $n$ vertices.
We obtain an algorithm running in time $\mathcal{O}\left(g(k)\cdot T\left(\ell\cdot n^{k/\ell}\right)\right)$ with $g(k) = (|\Fam|\cdot k^h 2^h)^{\ell \choose h}$. 
\end{proof}

\section{Further Application: Strengthening Known Reductions}
Another application that makes regularizing the hyperclique problem appealing is that
it might prove easier to show conditional hardness lower bounds when one can assume this strong regularity. 
This could not only be useful for new lower bounds, but also simplify existing proofs.

A point in case, where our technique obtains a lower bound for a regular version of the problem is the following:
Under the $3$-Uniform $4$-Hyperclique hypothesis it has been shown in~\cite{DalirrooyfardW22} that detecting induced $4$-cycles requires quadratic time for graphs with $\mathcal{O}(n^{3/2})$ edges.
Without any changes to the original reduction, we obtain the following by regularizing.
\begin{theorem}
    For no $\varepsilon > 0$ does there exist an algorithm detecting an induced $4$-cycle in a $r$-regular graph $G$ for $r \in \mathcal{O}(\sqrt n)$ (with $\mathcal{O}(n^{3/2})$ edges) in time $\mathcal{O}(n^{2 - \varepsilon})$ unless the $3$-uniform $k$-hyperclique hypothesis fails.
\end{theorem}
\begin{proof}
We shortly restate the procedure as described by Dalirrooyfard and {Vassilevska Williams} in \cite[Section 6]{DalirrooyfardW22} and only analyze the regularity of the resulting graph. 
To this end, let $G$ be a $4$-partite $3$-uniform $(2,\lambda)$-regular hypergraph, with parts $V_0,V_1,V_2,V_3$ of size $n$ each.

Now construct $G'$ with vertices $(x,y)$ such that $x \in V_i$ and $y \in V_{i+1}$ for $i \in \{0,1,2,3\}$ interpreting it modulo $4$ and the  following edges:
\begin{itemize}
    \item For every $i \in \{0,1,2,3\}$, between any two nodes $(x,y),(x',y') \in V_i \times V_{i+1}$ there is an edge if $y = y'$.
    \item For every $i \in\{0,1,2,3\}$, between any two nodes $(x,y) \in V_i \times V_{i+1}$ and $(x',y') \in V_{i+1} \times V_{i+2}$ there is an edge if $y = x'$ and $\{x,y,y'\}$ is a hyperedge in $G$.
\end{itemize}

To prove, that this construction indeed yields a regular graph, consider a fixed vertex $(x,y) \in V_i \times V_{i+1}$.
There are exactly $n-1$ edges $\{(x,y),(x',y)\}$, for every $(x',y) \in V_{i} \times V_{i+1}$ with $x' \neq x$.
As $G$ is $(2,\lambda)$-regular, there are exactly $\lambda$ many vertices $(y,z) \in V_{i+1} \times V_{i+2}$ such that we constructed an edge $\{(x,y),(y,z)\}$.
This holds analogously for $(z,x) \in V_{i-1} \times V_{i}$.

The argument as presented by Dalirrooyfard and {Vassilevska Williams} clearly still works out, such that $G'$ contains an induced $4$-cycle if and only if $G$ contains a $4$-clique.

The resulting graph has a total of $N = \bigO(n^2)$ vertices and $M = \bigO(n^3)$ edges.
Thus, $G'$ is an $r$-regular graph with $r = \mathcal{O}(n) = \mathcal{O}(\sqrt{N})$.
\end{proof}
Trying to regularize a graph $G$ without adding induced $4$-cycles by adding dummy vertices does not seem very straight forward.
While approaches, such as taken in \cite{Cai08}, do indeed regularize the given graph without adding new $4$-cycles by only adding edges between a dummy node and vertices that are not in the same component,
they might blow-up in size, as they are dependent on the maximum degree of the original graph.
Only adding dummy nodes on the other hand is not easily done, as connecting two non-adjacent vertices to a single dummy node might already introduce a $4$-cycle,
indicating that there is no obvious way that achieves regularity for this problem with only a small amount of dummy nodes added.

\bibliography{article}

\begin{thebibliography}{10}

\bibitem{AbboudBW15a}
Amir Abboud, Arturs Backurs, and Virginia {Vassilevska Williams}.
\newblock If the current clique algorithms are optimal, so is valiant's parser.
\newblock In Venkatesan Guruswami, editor, {\em {IEEE} 56th Annual Symposium on
  Foundations of Computer Science, {FOCS} 2015, Berkeley, CA, USA, 17-20
  October, 2015}, pages 98--117. {IEEE} Computer Society, 2015.
\newblock \href {https://doi.org/10.1109/FOCS.2015.16}
  {\path{doi:10.1109/FOCS.2015.16}}.

\bibitem{AbboudBDN18}
Amir Abboud, Karl Bringmann, Holger Dell, and Jesper Nederlof.
\newblock More consequences of falsifying {SETH} and the {Orthogonal Vectors}
  conjecture.
\newblock In {\em Proc. 50th Annual ACM SIGACT Symposium on Theory of Computing
  (STOC 2018)}, STOC 2018, pages 253--266, New York, NY, USA, 2018. ACM.
\newblock \href {https://doi.org/10.1145/3188745.3188938}
  {\path{doi:10.1145/3188745.3188938}}.

\bibitem{AlmanDWXXZ25}
Josh Alman, Ran Duan, Virginia {Vassilevska Williams}, Yinzhan Xu, Zixuan Xu,
  and Renfei Zhou.
\newblock More asymmetry yields faster matrix multiplication.
\newblock In Yossi Azar and Debmalya Panigrahi, editors, {\em Proceedings of
  the 2025 Annual {ACM-SIAM} Symposium on Discrete Algorithms, {SODA} 2025, New
  Orleans, LA, USA, January 12-15, 2025}, pages 2005--2039. {SIAM}, 2025.
\newblock \href {https://doi.org/10.1137/1.9781611978322.63}
  {\path{doi:10.1137/1.9781611978322.63}}.

\bibitem{AnGIJKN21}
Haozhe An, Mohit Gurumukhani, Russell Impagliazzo, Michael Jaber, Marvin
  K{\"{u}}nnemann, and Maria Paula~Parga Nina.
\newblock The fine-grained complexity of multi-dimensional ordering properties.
\newblock In Petr~A. Golovach and Meirav Zehavi, editors, {\em 16th
  International Symposium on Parameterized and Exact Computation, {IPEC} 2021,
  September 8-10, 2021, Lisbon, Portugal}, volume 214 of {\em LIPIcs}, pages
  3:1--3:15. Schloss Dagstuhl - Leibniz-Zentrum f{\"{u}}r Informatik, 2021.
\newblock URL: \url{https://doi.org/10.4230/LIPIcs.IPEC.2021.3}, \href
  {https://doi.org/10.4230/LIPICS.IPEC.2021.3}
  {\path{doi:10.4230/LIPICS.IPEC.2021.3}}.

\bibitem{BrandesHK16}
Ulrik Brandes, Eugenia Holm, and Andreas Karrenbauer.
\newblock Cliques in regular graphs and the core-periphery problem in social
  networks.
\newblock In T.{-}H.~Hubert Chan, Minming Li, and Lusheng Wang, editors, {\em
  Combinatorial Optimization and Applications - 10th International Conference,
  {COCOA} 2016, Hong Kong, China, December 16-18, 2016, Proceedings}, volume
  10043 of {\em Lecture Notes in Computer Science}, pages 175--186. Springer,
  2016.
\newblock \href {https://doi.org/10.1007/978-3-319-48749-6\_13}
  {\path{doi:10.1007/978-3-319-48749-6\_13}}.

\bibitem{BringmannCFK22}
Karl Bringmann, Alejandro Cassis, Nick Fischer, and Marvin K{\"{u}}nnemann.
\newblock A structural investigation of the approximability of polynomial-time
  problems.
\newblock In Mikolaj Bojanczyk, Emanuela Merelli, and David~P. Woodruff,
  editors, {\em 49th International Colloquium on Automata, Languages, and
  Programming, {ICALP} 2022, July 4-8, 2022, Paris, France}, volume 229 of {\em
  LIPIcs}, pages 30:1--30:20. Schloss Dagstuhl - Leibniz-Zentrum f{\"{u}}r
  Informatik, 2022.
\newblock URL: \url{https://doi.org/10.4230/LIPIcs.ICALP.2022.30}, \href
  {https://doi.org/10.4230/LIPICS.ICALP.2022.30}
  {\path{doi:10.4230/LIPICS.ICALP.2022.30}}.

\bibitem{BringmannFK19}
Karl Bringmann, Nick Fischer, and Marvin K{\"{u}}nnemann.
\newblock A fine-grained analogue of schaefer's theorem in {P:} dichotomy of
  exists{\^{}}k-forall-quantified first-order graph properties.
\newblock In Amir Shpilka, editor, {\em 34th Computational Complexity
  Conference, {CCC} 2019, July 18-20, 2019, New Brunswick, NJ, {USA}}, volume
  137 of {\em LIPIcs}, pages 31:1--31:27. Schloss Dagstuhl - Leibniz-Zentrum
  f{\"{u}}r Informatik, 2019.
\newblock URL: \url{https://doi.org/10.4230/LIPIcs.CCC.2019.31}, \href
  {https://doi.org/10.4230/LIPICS.CCC.2019.31}
  {\path{doi:10.4230/LIPICS.CCC.2019.31}}.

\bibitem{BringmannS21}
Karl Bringmann and Jasper Slusallek.
\newblock Current algorithms for detecting subgraphs of bounded treewidth are
  probably optimal.
\newblock In Nikhil Bansal, Emanuela Merelli, and James Worrell, editors, {\em
  48th International Colloquium on Automata, Languages, and Programming,
  {ICALP} 2021, July 12-16, 2021, Glasgow, Scotland (Virtual Conference)},
  volume 198 of {\em LIPIcs}, pages 40:1--40:16. Schloss Dagstuhl -
  Leibniz-Zentrum f{\"{u}}r Informatik, 2021.
\newblock URL: \url{https://doi.org/10.4230/LIPIcs.ICALP.2021.40}, \href
  {https://doi.org/10.4230/LIPICS.ICALP.2021.40}
  {\path{doi:10.4230/LIPICS.ICALP.2021.40}}.

\bibitem{BringmannW17}
Karl Bringmann and Philip Wellnitz.
\newblock Clique-based lower bounds for parsing tree-adjoining grammars.
\newblock In Juha K{\"{a}}rkk{\"{a}}inen, Jakub Radoszewski, and Wojciech
  Rytter, editors, {\em 28th Annual Symposium on Combinatorial Pattern
  Matching, {CPM} 2017, July 4-6, 2017, Warsaw, Poland}, volume~78 of {\em
  LIPIcs}, pages 12:1--12:14. Schloss Dagstuhl - Leibniz-Zentrum f{\"{u}}r
  Informatik, 2017.
\newblock URL: \url{https://doi.org/10.4230/LIPIcs.CPM.2017.12}, \href
  {https://doi.org/10.4230/LIPICS.CPM.2017.12}
  {\path{doi:10.4230/LIPICS.CPM.2017.12}}.

\bibitem{Bulatov17}
Andrei~A. Bulatov.
\newblock A dichotomy theorem for nonuniform {CSPs}.
\newblock In {\em Proc. 58th {IEEE} Annual Symposium on Foundations of Computer
  Science ({FOCS} 2017)}, pages 319--330, 2017.
\newblock \href {https://doi.org/10.1109/FOCS.2017.37}
  {\path{doi:10.1109/FOCS.2017.37}}.

\bibitem{BulatovM14}
Andrei~A. Bulatov and D{\'{a}}niel Marx.
\newblock Constraint satisfaction parameterized by solution size.
\newblock {\em {SIAM} J. Comput.}, 43(2):573--616, 2014.
\newblock \href {https://doi.org/10.1137/120882160}
  {\path{doi:10.1137/120882160}}.

\bibitem{Cai08}
Leizhen Cai.
\newblock Parameterized complexity of cardinality constrained optimization
  problems.
\newblock {\em Comput. J.}, 51(1):102--121, 2008.
\newblock URL: \url{https://doi.org/10.1093/comjnl/bxm086}, \href
  {https://doi.org/10.1093/COMJNL/BXM086} {\path{doi:10.1093/COMJNL/BXM086}}.

\bibitem{CarmeliZBKS20}
Nofar Carmeli, Shai Zeevi, Christoph Berkholz, Benny Kimelfeld, and Nicole
  Schweikardt.
\newblock Answering (unions of) conjunctive queries using random access and
  random-order enumeration.
\newblock In Dan Suciu, Yufei Tao, and Zhewei Wei, editors, {\em Proceedings of
  the 39th {ACM} {SIGMOD-SIGACT-SIGAI} Symposium on Principles of Database
  Systems, {PODS} 2020, Portland, OR, USA, June 14-19, 2020}, pages 393--409.
  {ACM}, 2020.
\newblock \href {https://doi.org/10.1145/3375395.3387662}
  {\path{doi:10.1145/3375395.3387662}}.

\bibitem{Chang16}
Yi{-}Jun Chang.
\newblock Hardness of {RNA} folding problem with four symbols.
\newblock In Roberto Grossi and Moshe Lewenstein, editors, {\em 27th Annual
  Symposium on Combinatorial Pattern Matching, {CPM} 2016, June 27-29, 2016,
  Tel Aviv, Israel}, volume~54 of {\em LIPIcs}, pages 13:1--13:12. Schloss
  Dagstuhl - Leibniz-Zentrum f{\"{u}}r Informatik, 2016.
\newblock URL: \url{https://doi.org/10.4230/LIPIcs.CPM.2016.13}, \href
  {https://doi.org/10.4230/LIPICS.CPM.2016.13}
  {\path{doi:10.4230/LIPICS.CPM.2016.13}}.

\bibitem{ColbournD06}
Charles~J. Colbourn and Jeffrey~H. Dinitz, editors.
\newblock {\em Handbook of {Combinatorial} {Designs}}.
\newblock Chapman and Hall/CRC, New York, 2 edition, November 2006.
\newblock \href {https://doi.org/10.1201/9781420010541}
  {\path{doi:10.1201/9781420010541}}.

\bibitem{Creignou95}
Nadia Creignou.
\newblock A dichotomy theorem for maximum generalized satisfiability problems.
\newblock {\em J. Comput. Syst. Sci.}, 51(3):511--522, 1995.
\newblock \href {https://doi.org/10.1006/jcss.1995.1087}
  {\path{doi:10.1006/jcss.1995.1087}}.

\bibitem{CyganFKLMPPS15}
Marek Cygan, Fedor~V. Fomin, Lukasz Kowalik, Daniel Lokshtanov, D{\'{a}}niel
  Marx, Marcin Pilipczuk, Michal Pilipczuk, and Saket Saurabh.
\newblock {\em Parameterized Algorithms}.
\newblock Springer, 2015.
\newblock \href {https://doi.org/10.1007/978-3-319-21275-3}
  {\path{doi:10.1007/978-3-319-21275-3}}.

\bibitem{DalirrooyfardW22}
Mina Dalirrooyfard and Virginia {Vassilevska Williams}.
\newblock Induced cycles and paths are harder than you think.
\newblock In {\em 63rd {IEEE} Annual Symposium on Foundations of Computer
  Science, {FOCS} 2022, Denver, CO, USA, October 31 - November 3, 2022}, pages
  531--542. {IEEE}, 2022.
\newblock \href {https://doi.org/10.1109/FOCS54457.2022.00057}
  {\path{doi:10.1109/FOCS54457.2022.00057}}.

\bibitem{Feige98}
Uriel Feige.
\newblock A threshold of ln \emph{n} for approximating set cover.
\newblock {\em J. {ACM}}, 45(4):634--652, 1998.
\newblock \href {https://doi.org/10.1145/285055.285059}
  {\path{doi:10.1145/285055.285059}}.

\bibitem{GorbachevK23}
Egor Gorbachev and Marvin K{\"{u}}nnemann.
\newblock Combinatorial designs meet hypercliques: Higher lower bounds for
  klee's measure problem and related problems in dimensions d {\(\geq\)} 4.
\newblock In Erin~W. Chambers and Joachim Gudmundsson, editors, {\em 39th
  International Symposium on Computational Geometry, SoCG 2023, June 12-15,
  2023, Dallas, Texas, {USA}}, volume 258 of {\em LIPIcs}, pages 36:1--36:14.
  Schloss Dagstuhl - Leibniz-Zentrum f{\"{u}}r Informatik, 2023.
\newblock URL: \url{https://doi.org/10.4230/LIPIcs.SoCG.2023.36}, \href
  {https://doi.org/10.4230/LIPICS.SOCG.2023.36}
  {\path{doi:10.4230/LIPICS.SOCG.2023.36}}.

\bibitem{KhannaSTW00}
Sanjeev Khanna, Madhu Sudan, Luca Trevisan, and David~P. Williamson.
\newblock The approximability of constraint satisfaction problems.
\newblock {\em {SIAM} J. Comput.}, 30(6):1863--1920, 2000.
\newblock \href {https://doi.org/10.1137/S0097539799349948}
  {\path{doi:10.1137/S0097539799349948}}.

\bibitem{KratschMW16}
Stefan Kratsch, D{\'{a}}niel Marx, and Magnus Wahlstr{\"{o}}m.
\newblock Parameterized complexity and kernelizability of max ones and exact
  ones problems.
\newblock {\em {TOCT}}, 8(1):1:1--1:28, 2016.
\newblock \href {https://doi.org/10.1145/2858787} {\path{doi:10.1145/2858787}}.

\bibitem{Kunnemann22}
Marvin K{\"{u}}nnemann.
\newblock A tight (non-combinatorial) conditional lower bound for klee's
  measure problem in 3d.
\newblock In {\em 63rd {IEEE} Annual Symposium on Foundations of Computer
  Science, {FOCS} 2022, Denver, CO, USA, October 31 - November 3, 2022}, pages
  555--566. {IEEE}, 2022.
\newblock \href {https://doi.org/10.1109/FOCS54457.2022.00059}
  {\path{doi:10.1109/FOCS54457.2022.00059}}.

\bibitem{KunnemannM20}
Marvin K{\"{u}}nnemann and D{\'{a}}niel Marx.
\newblock Finding small satisfying assignments faster than brute force: {A}
  fine-grained perspective into boolean constraint satisfaction.
\newblock In Shubhangi Saraf, editor, {\em 35th Computational Complexity
  Conference, {CCC} 2020, July 28-31, 2020, Saarbr{\"{u}}cken, Germany (Virtual
  Conference)}, volume 169 of {\em LIPIcs}, pages 27:1--27:28. Schloss Dagstuhl
  - Leibniz-Zentrum f{\"{u}}r Informatik, 2020.
\newblock URL: \url{https://doi.org/10.4230/LIPIcs.CCC.2020.27}, \href
  {https://doi.org/10.4230/LIPICS.CCC.2020.27}
  {\path{doi:10.4230/LIPICS.CCC.2020.27}}.

\bibitem{LincolnWW18}
Andrea Lincoln, Virginia {Vassilevska Williams}, and R.~Ryan Williams.
\newblock Tight hardness for shortest cycles and paths in sparse graphs.
\newblock In Artur Czumaj, editor, {\em Proceedings of the Twenty-Ninth Annual
  {ACM-SIAM} Symposium on Discrete Algorithms, {SODA} 2018, New Orleans, LA,
  USA, January 7-10, 2018}, pages 1236--1252. {SIAM}, 2018.
\newblock \href {https://doi.org/10.1137/1.9781611975031.80}
  {\path{doi:10.1137/1.9781611975031.80}}.

\bibitem{Marx05}
D{\'{a}}niel Marx.
\newblock Parameterized complexity of constraint satisfaction problems.
\newblock {\em Computational Complexity}, 14(2):153--183, 2005.
\newblock \href {https://doi.org/10.1007/s00037-005-0195-9}
  {\path{doi:10.1007/s00037-005-0195-9}}.

\bibitem{Marx06}
D{\'{a}}niel Marx.
\newblock Parameterized graph separation problems.
\newblock {\em Theor. Comput. Sci.}, 351(3):394--406, 2006.
\newblock URL: \url{https://doi.org/10.1016/j.tcs.2005.10.007}, \href
  {https://doi.org/10.1016/J.TCS.2005.10.007}
  {\path{doi:10.1016/J.TCS.2005.10.007}}.

\bibitem{MathiesonS08}
Luke Mathieson and Stefan Szeider.
\newblock The parameterized complexity of regular subgraph problems and
  generalizations.
\newblock In James Harland and Prabhu Manyem, editors, {\em Theory of Computing
  2008. Proc. Fourteenth Computing: The Australasian Theory Symposium {(CATS}
  2008), Wollongong, NSW, Australia, January 22-25, 2008. Proceedings},
  volume~77 of {\em {CRPIT}}, pages 79--86. Australian Computer Society, 2008.
\newblock URL:
  \url{http://crpit.scem.westernsydney.edu.au/abstracts/CRPITV77Mathieson.html}.

\bibitem{NesetrilP85}
Jaroslav Ne\v{s}et\v{r}il and Svatopluk Poljak.
\newblock On the complexity of the subgraph problem.
\newblock {\em Commentationes Mathematicae Universitatis Carolinae},
  026(2):415--419, 1985.

\bibitem{Schaefer78}
Thomas~J. Schaefer.
\newblock The complexity of satisfiability problems.
\newblock In {\em Proceedings of the 10th Annual {ACM} Symposium on Theory of
  Computing, May 1-3, 1978, San Diego, California, {USA}}, pages 216--226,
  1978.
\newblock \href {https://doi.org/10.1145/800133.804350}
  {\path{doi:10.1145/800133.804350}}.

\bibitem{VassilevskaW06}
Virginia Vassilevska and Ryan Williams.
\newblock Finding a maximum weight triangle in n\({}^{\mbox{3-delta}}\) time,
  with applications.
\newblock In Jon~M. Kleinberg, editor, {\em Proceedings of the 38th Annual
  {ACM} Symposium on Theory of Computing, Seattle, WA, USA, May 21-23, 2006},
  pages 225--231. {ACM}, 2006.
\newblock \href {https://doi.org/10.1145/1132516.1132550}
  {\path{doi:10.1145/1132516.1132550}}.

\bibitem{Vassilevska18}
Virginia {Vassilevska Williams}.
\newblock On some fine-grained questions in algorithms and complexity.
\newblock In {\em Proceedings of the international congress of mathematicians:
  Rio de janeiro 2018}, pages 3447--3487, 2018.

\bibitem{Wil07}
R.~Ryan Williams.
\newblock {\em Algorithms and resource requirements for fundamental problems}.
\newblock PhD thesis, Carnegie Mellon University, USA, 2007.
\newblock AAI3274191.

\bibitem{Zamir23}
Or~Zamir.
\newblock Algorithmic applications of hypergraph and partition containers.
\newblock In Barna Saha and Rocco~A. Servedio, editors, {\em Proceedings of the
  55th Annual {ACM} Symposium on Theory of Computing, {STOC} 2023, Orlando, FL,
  USA, June 20-23, 2023}, pages 985--998. {ACM}, 2023.
\newblock \href {https://doi.org/10.1145/3564246.3585163}
  {\path{doi:10.1145/3564246.3585163}}.

\bibitem{Zhuk17}
Dmitriy Zhuk.
\newblock A proof of {CSP} dichotomy conjecture.
\newblock In {\em Proc. 58th {IEEE} Annual Symposium on Foundations of Computer
  Science ({FOCS} 2017)}, pages 331--342, 2017.
\newblock \href {https://doi.org/10.1109/FOCS.2017.38}
  {\path{doi:10.1109/FOCS.2017.38}}.

\bibitem{Zhuk20}
Dmitriy Zhuk.
\newblock A proof of the {CSP} dichotomy conjecture.
\newblock {\em J. {ACM}}, 67(5):30:1--30:78, 2020.
\newblock \href {https://doi.org/10.1145/3402029} {\path{doi:10.1145/3402029}}.

\bibitem{Zwick02}
Uri Zwick.
\newblock All pairs shortest paths using bridging sets and rectangular matrix
  multiplication.
\newblock {\em J. {ACM}}, 49(3):289--317, 2002.
\newblock \href {https://doi.org/10.1145/567112.567114}
  {\path{doi:10.1145/567112.567114}}.

\end{thebibliography}

\end{document}